\documentclass[12pt, draftclsnofoot, onecolumn]{IEEEtran}
\usepackage[utf8]{inputenc}

\usepackage{tabularx}
\usepackage{makecell}

\usepackage[table,xcdraw]{xcolor}
\usepackage{threeparttable}
\usepackage{algpseudocode}
\usepackage{setspace}

\usepackage{stfloats}
\usepackage{amsfonts}
\usepackage{amssymb}
\usepackage{amsthm}
\usepackage{cite}
\usepackage[cmex10]{amsmath}
\usepackage{float}
\usepackage{color}
\usepackage{stfloats,fancyhdr}
\usepackage{amsmath,bm}
\usepackage{algorithm}
\usepackage{multirow}
\usepackage{changepage}
\usepackage[normalem]{ulem}
\usepackage{balance}
\usepackage{bbm}
\usepackage[font=footnotesize]{caption}

\newtheorem{problem}{Problem}
\newtheorem{theorem}{Theorem}
\newtheorem{lemma}{Lemma}
\newtheorem{proposition}{Proposition}

\newtheorem{definition}{Definition}

\newtheorem{remark}{Remark}

\ifCLASSINFOpdf
\usepackage[pdftex]{graphicx}
\DeclareGraphicsExtensions{.pdf,.jpeg,.png}
\else
\usepackage[dvips]{graphicx}
\DeclareGraphicsExtensions{.eps}
\fi

\usepackage{subfigure}
\usepackage{fancybox,dashbox}
\usepackage{authblk}
\usepackage{tabularx}

\usepackage{mathtools}
\mathtoolsset{showonlyrefs=true}

\title{Structure-Enhanced DRL for Optimal Transmission Scheduling}
\author{Jiazheng Chen, Wanchun Liu*,~\IEEEmembership{Member,~IEEE,} Daniel~E.~Quevedo,~\IEEEmembership{Fellow,~IEEE,} Saeed R. Khosravirad,~\IEEEmembership{Member,~IEEE,} Yonghui~Li,~\IEEEmembership{Fellow,~IEEE,} Branka Vucetic,~\IEEEmembership{Life Fellow,~IEEE}

\thanks{Part of the work has been submitted to IEEE ICC 2023 \cite{Chen2023SEDRL}.}
\thanks{J. Chen, W. Liu, Y. Li, and B. Vucetic are with the School of Electrical and Information Engineering, The University of Sydney, Sydney, NSW 2006, Australia (e-mail: jiazheng.chen@sydney.edu.au; wanchun.liu@sydney.edu.au; yonghui.li@sydney.edu.au; branka.vucetic@sydney.edu.au).
D. E. Quevedo is with the School of Electrical Engineering and Robotics, Queensland University of Technology (QUT), Brisbane, Australia. (e-mail: dquevedo@ieee.org).
S. R. Khosravirad is with the Nokia Bell Laboratories, Murry Hill, NJ 07964 USA (e-mail: saeed.khosravirad@nokia-bell-labs.com).
\textit{(W. Liu is the corresponding author.)}}} 
\date{September 2022}

\begin{document}

\maketitle
\vspace{-1.8cm}

\begin{abstract}
    Remote state estimation of large-scale distributed dynamic processes plays an important role in Industry 4.0 applications. In this paper, we focus on the transmission scheduling problem of a remote estimation system. First, we derive some structural properties of the optimal sensor scheduling policy over fading channels. Then, building on these theoretical guidelines, we develop a structure-enhanced deep reinforcement learning (DRL) framework for optimal scheduling of the system to achieve the minimum overall estimation mean-square error (MSE).
    In particular, we propose a structure-enhanced action selection method, which tends to select actions that obey the policy structure. This explores the action space more effectively and enhances the learning efficiency of DRL agents.
    Furthermore, we introduce a structure-enhanced loss function to add penalties to actions that do not follow the policy structure. The new loss function guides the DRL to converge to the optimal policy structure quickly.
    Our numerical experiments illustrate that the proposed structure-enhanced DRL algorithms can save the training time by 50\% and reduce the remote estimation MSE by 10\% to 25\%, when compared to benchmark DRL algorithms. 
    In addition, we show that the derived structural properties exist in a wide range of dynamic scheduling problems that go beyond remote state estimation.
\end{abstract}
\vspace{-0.5cm}

\vspace{-0.5cm}

\section{Introduction} \label{sec:intro}
Wireless networked control systems (WNCSs) consisting of distributed sensors, actuators, controllers, and plants are a key component for Industry 4.0 and have been widely applied in many areas such as industrial automation, vehicle monitoring systems, building automation and smart grids \cite{Park2018WNCS}. In particular, providing high-quality real-time remote estimation of dynamic system plant states plays an important role in ensuring control performance and stability of WNCSs \cite{Huang2020UpDown, Schenato2007Stability}.
For large-scale WNCSs, transmission scheduling of wireless sensors over limited bandwidth needs to be properly designed to guarantee remote estimation performance.

There are many existing works on the transmission scheduling of WNCSs. 
In \cite{gatsis2015opportunistic}, the optimal scheduling problem of a multi-loop WNCS with limited communication resources was investigated for minimizing the transmission power consumption under a WNCS stability constraint.
In~\cite{han2017optimal,wu2018optimalmulti}, optimal sensor scheduling problems of remote estimation systems were investigated for achieving the best overall estimation mean-square error (MSE). In particular, dynamic decision-making problems were formulated as Markov decision processes (MDPs), which can be solved by classical  methods, such as policy and value iterations.
The recent work~\cite{Forootani2022NOMA} reduces the computation complexity for solving optimal scheduling problems through approximate dynamic programming.
However, the conventional model-based solutions are not feasible in large-scale scheduling problems because of the curse of dimensionality caused by high-dimensional state and action spaces.\footnote{In~\cite{gatsis2015opportunistic,han2017optimal,wu2018optimalmulti}, only the optimal scheduling of two-sensor systems has been solved effectively by the proposed methods.} 

In recent years, deep reinforcement learning (DRL) has been developed to deal with large MDPs by using deep neural networks as function approximators \cite{sutton2018reinforcement, zhao2022deep}. 
Some works \cite{leong2020DRL, liu2021drlscheduling, Demirel2018DEEPCAS, Yang2022OMA} have used the deep Q-network (DQN), a simple DRL method, to solve multi-sensor-multi-channel scheduling problems in different remote estimation scenarios. In particular, sensor scheduling problems for systems with 6 sensors have been solved effectively, providing significant performance gains over heuristic methods in terms of estimation quality. The more recent work~\cite{pang2022drl} has introduced DRL algorithms with an actor-critic structure to solve scheduling problems at a much larger scale (that cannot be handled by the DQN). However, existing works merely use the general DRL frameworks to solve specific scheduling problems, without questioning what features distinguish sensor scheduling problems from other MDPs. Also, we note that a drawback of general DRL is that it often cannot perform policy exploration effectively for specific tasks~\cite{guo2019exploration}, which can lead to getting stuck in local minima or even total failure. Thus, the existing DRL-based solutions could be far from optimal.

A key feature that we will exploit in our current work is that optimal transmission scheduling policies of remote estimation systems often have threshold structures~\cite{wu2019single,wu2020optimalmulti,wu2018optimalmulti}, which means that there exist switching boundaries dividing the state space into multiple regions for different scheduling actions. In other words, an optimal policy has a structure where the action only changes at the switching boundaries of the state space.
In particular, \cite{wu2019single} focuses on an energy-constrained single-sensor-single-channel system and proves that the optimal policy has a threshold in terms of the sensor's age of information (AoI), determining whether the sensor will be scheduled or not.
In~\cite{wu2020optimalmulti}, the authors considered a multi-sensor-multi-channel system, where all channels are static, and each sensor has a constant packet-drop probability at all frequency channels. This work also showed that the optimal scheduling policy has a multi-dimensional threshold structure in terms of all sensor AoI. 
As an extension, the work in~\cite{wu2018optimalmulti} (a two-sensor system), assumed that different sensors could have different numbers of packets for carrying each measurement. The demonstrated threshold property of the optimal policy is related to the sensor AoI and the remaining packet numbers of each sensor. 
There are two limitations of the channel models adopted in the above works: 1) fading channel models are commonly adopted in practice, where channel states are time-varying, and 2) wireless propagation via different frequency bandwidths has different properties, leading to different channel qualities.
Not limited to remote estimation systems, in~\cite{hsu2017threshold}, the threshold structure of an optimal sensor scheduling policy has also been identified for minimizing the average sum AoI. However, this work only considered a single-channel system, and the transmission success or failure was determined before a scheduling action.
Therefore, the theoretical works~\cite{wu2019single,wu2020optimalmulti,wu2018optimalmulti,hsu2017threshold} only derived the structural properties of optimal policies under some ideal assumptions. It is an open challenge to prove the existence of structural results of the optimal scheduling policy of a general multi-sensor-multi-channel system over practical fading channels.
Furthermore, there is no existing work in the open literature utilizing the structural properties to guide DRL algorithms for effectively solving optimal scheduling problems.

In this paper, we consider the optimal sensor scheduling problem of a general $N$-sensor-$M$-channel remote estimation system over fading channels. The main contributions of this work are summarized as follows.
\begin{itemize}
    \item[$\bullet$] We prove that the optimal sensor scheduling policy has a threshold structure in terms of both the AoI states of all sensors and the corresponding channel states, where the channel states of different sensors at different frequencies are different. To the best of our knowledge, this is the first structural result of optimal scheduling policies over fading channels in the literature. In addition, we show that such a structural property exists in a wide range of dynamic scheduling/resource allocation problems, not limited to remote state estimations. Interestingly, we also give a counterexample to show when such a property does not exist.
    \item[$\bullet$] We first formulate the sensor scheduling problem into an MDP, and then develop novel structure-enhanced DRL algorithms for solving the problem, building on the derived threshold properties of the optimal policy. 
    In particular, we design a \textbf{structure-enhanced action selection method}, which tends to select actions that obey the threshold structure. Such an action selection method can explore the action space more effectively and enhance the learning efficiency of DRL agents.
    Furthermore, we introduce a \textbf{structure-enhanced loss function} to add penalties to actions that do not follow the threshold structure. The new loss function guides the DRL to converge to the optimal policy structure quickly. 
    We apply the proposed action selection method and the novel loss function to redesign the most commonly adopted DRL frameworks for scheduling, i.e., DQN and deep deterministic policy gradient (DDPG), referred to as the \textbf{structure-enhanced DQN and DDPG algorithms}.
    \item[$\bullet$] Our extensive numerical results illustrate that the proposed structure-enhanced DRL algorithms can save the training time by 50\% while reducing the remote estimation MSE by 10\% to 25\% compared with benchmark DRL algorithms. Importantly, the structure-enhanced DRL algorithms can converge and perform well under some system settings that cannot be solved by any of the benchmark DRL algorithms.
\end{itemize}
\begin{figure}[t]
    \centering
    \includegraphics[width=0.65\linewidth]{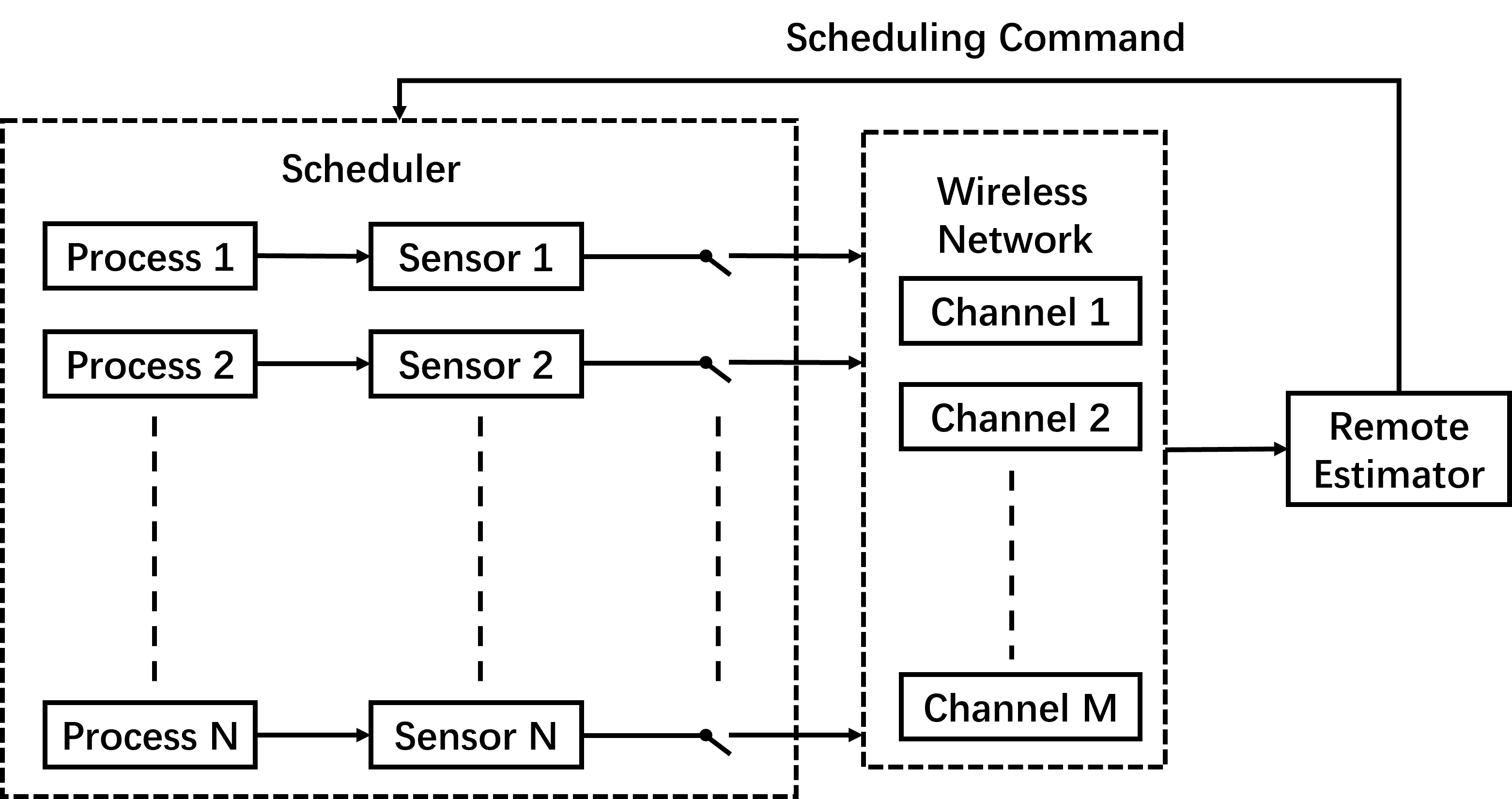}
    \vspace{-0.2cm}
    \caption{Remote state estimation system with $N$ processes and $M$ channels.}
    \label{fig:system_model}
    \vspace{-0.7cm}
\end{figure}
\textbf{Outline:} The system model of the remote state estimation system is described in Section~\ref{sec: sys}. The sensor scheduling problem formulation and the introduction of the structural properties are presented in Section~\ref{sec: problem}. The proofs of structural properties of the optimal scheduling policy are given in Section~\ref{sec: proof of threshold structure}. The structure-enhanced DRL algorithms for the formulated problem are presented in Section~\ref{sec: DRL}. The numerical results are shown and analyzed in Section~\ref{sec: simulation}, followed by conclusions in Section~\ref{sec: conclusion}.

\section{System Model} \label{sec: sys}
We consider a remote estimation system with $N$ dynamic processes, each measured by a sensor, which pre-processes the raw measurements and sends its state estimates to a remote estimator through one of $M$ wireless channels, as illustrated in Fig.~\ref{fig:system_model}.

\subsection{Dynamic Process Model and Local State Estimation}
Each dynamic process $n$ is modeled as a discrete-time linear time-invariant (LTI) system as \cite{leong2020DRL,liu2021remoteMF, liu2022stability}
\begin{equation}\label{eq:LTI}
\begin{aligned}
    \mathbf{x}_{n,t+1} & = \mathbf{A}_{n} \mathbf{x}_{n, t} + \mathbf{w}_{n, t},  \\
    \mathbf{y}_{n, t} & = \mathbf{C}_{n} \mathbf{x}_{n, t} + \mathbf{v}_{n, t}, n\in\{1,\dots,N\}, t\in\{1,\dots\},
\end{aligned}
\end{equation}
where $\mathbf{x}_{n, t} \in \mathbb{R}^{l_{n}}$ is  process $n$'s state at time $t$, and  $\mathbf{y}_{n, t} \in \mathbb{R}^{e_{n}}$ is the state measurement of the sensor $n$, 
$\mathbf{A}_{n} \in \mathbb{R}^{l_{n} \times l_{n}}$ and $\mathbf{C}_{n} \in \mathbb{R}^{e_{n} \times l_{n}}$ are the system matrix
and the measurement matrix, respectively, $\mathbf{w}_{n, t} \in \mathbb{R}^{l_{n}}$ and $\mathbf{v}_{n, t} \in \mathbb{R}^{e_{n}}$ are the process disturbance and the measurement noise modeled as independent and identically distributed (i.i.d) zero-mean Gaussian random vectors $\mathcal{N}(\mathbf{0},\mathbf{W}_{n})$ and $\mathcal{N}(\mathbf{0},\mathbf{V}_{n})$, respectively. \emph{We assume that the spectral radius of $\mathbf{A}_n, \forall n$, is greater than one, which means that the dynamic processes are unstable, making the remote estimation problem more interesting} (see \cite{liu2021remoteMF} and references therein).

Due to the presence of noise in \eqref{eq:LTI}, each sensor $n$ executes a classic Kalman filter to pre-process the raw measurement and generate state estimate $\mathbf{x}^s_{n,t}$  at each time $t$~\cite{liu2021remoteMF} as
\begin{subequations}
    \begin{align}
    \mathbf{x}_{n, t|t-1}^{s} & = \mathbf{A}_{n}\mathbf{x}_{n, t-1}^{s} \label{KF,a} \\
    \mathbf{P}_{n, t|t-1}^{s} & = \mathbf{A}_{n} \mathbf{P}_{n, t-1}^{s} \mathbf{A}_{n}^{\top} + \mathbf{W}_{n} \label{KF,b} \\
    \mathbf{K}_{n, t} & = \mathbf{P}_{n, t|t-1}^{s} \mathbf{C}_{n}^{\top} (\mathbf{C}_{n} \mathbf{P}_{n, t|t-1}^{s} \mathbf{C}_{n}^{\top}+\mathbf{V}_{n})^{-1} \label{KF,c} \\
    \mathbf{x}_{n, t}^{s} & = \mathbf{x}_{n, t|t-1}^{s} + \mathbf{K}_{n, t} (\mathbf{y}_{n, t} - \mathbf{C}\mathbf{x}_{n, t|t-1}^{s}) \label{KF,d} \\
    \mathbf{P}_{n, t}^{s} & = (\mathbf{I}_{n} - \mathbf{K}_{n, t} \mathbf{C}_{n}) \mathbf{P}_{n, t|t-1}^{s} \label{KF,e}
    \end{align}
\end{subequations}
where $\mathbf{x}_{n, t|t-1}^{s}$ and $\mathbf{P}_{n, t|t-1}^{s}$ are the prior state estimate and the corresponding estimation error covariance of sensor $n$, respectively, $\mathbf{x}_{n, t}^{s}$ and $\mathbf{P}_{n, t}^{s}$ are the posterior state estimate and the corresponding estimation error covariance of sensor $n$ at time $t$, respectively. 
In particular, sensor $n$ sends the estimate $\mathbf{x}^s_{n,t}$ to the remote estimator (not $\mathbf{y}_{n,t}$) as a packet, once scheduled, and the local state estimation error covariance matrix is defined as 
\begin{equation}\label{eq:local_P}
    \mathbf{P}_{n, t}^{s} \triangleq \mathbb{E} \left[ (\mathbf{x}_{n, t}^{s} - \mathbf{x}_{n, t}) (\mathbf{x}_{n, t}^{s} - \mathbf{x}_{n, t})^{\top} \right].
\end{equation}
$\mathbf{K}_{n, t}$ is the Kalman gain of sensor $n$, and $\mathbf{I}_{n}$ is an identity matrix. Note that~\eqref{KF,a} and~\eqref{KF,b} present the prediction steps while~\eqref{KF,c},~\eqref{KF,d}, and~\eqref{KF,e} are the updating steps.
We note that local Kalman filters are commonly assumed to operate in the steady state mode in the literature (see \cite{liu2021remoteMF} and references therein\footnote{The $n$th local Kalman filter converges to  steady state if $(\mathbf{A}_{n},\mathbf{C}_{n})$ is observable and $(\mathbf{A}_{n}, \sqrt{\mathbf{W}_{n}})$ is controllable.}). We thus assume that the error covariance matrix has converged to a constant, i.e.,  $\mathbf{P}_{n, t}^{s} = \bar{\mathbf{P}}_{n}, \forall t,n$.

\subsection{Wireless Communications and Remote State Estimation}
There are only $M$ wireless channels (e.g., subcarriers) for the $N$ sensors' transmissions, where $M \leq N$.
We consider independent and identically distributed (i.i.d.) block fading channels, where the channel quality is fixed during each packet transmission and varies packet by packet, independently.
Let the $N\times M$ matrix $\mathbf{H}_{t}$ denote the channel state of the system at time $t$, where the element in the $n$th row and $m$th column, say  $h_{n,m,t} \in \mathcal{H} \triangleq \left\{ 1, 2, \dots, \bar{h} \right\}$, represents the channel state between sensor $n$ and the remote estimator at channel $m$. 
In particular, there are $\bar{h}$ quantized channel states in total. 
The distribution of $h_{n,m,t}$ is given~as
\begin{equation}\label{eq:q}
\operatorname{Pr}(h_{n,m,t}=i) = q^{(n,m)}_{i}, \forall t,
\end{equation}
where $\sum_{i=1}^{\bar{h}} q^{(n,m)}_{i} =1, \forall n,m$.
The instantaneous channel state $\mathbf{H}_{t}$ is available at the remote estimator based on standard channel estimation methods.

The packet drop probability at channel state $i' \in \mathcal{H}$ is denoted as $\tilde{p}_{i'}$. Without loss of generality, we assume that $\tilde{p}_{1} \geq \tilde{p}_{2} \geq \dots \geq \tilde{p}_{\bar{h}}$. We also define the packet success rate for the channel state $h_{n,m,t}$ as $p_{n,m,t}\in \{p_1,\dots,p_{\bar{h}}\}$, where $p_{i'} \triangleq 1-\tilde{p}_{i'}$.

Due to the limited communication channels, only $M$ out of $N$ sensors can be scheduled at each time step. Let $a_{n,t} \in \{0, 1, 2, \dots, M\}$ represent the channel allocation for sensor $n$ at time $t$, where
\begin{equation}\label{eq:action}
    a_{n,t} = \left\{
    \begin{array}{ll}
        0  & \text{if sensor $n$ is not scheduled} \vspace{-0.1cm}\\
        m  & \text{if sensor $n$ is scheduled to channel $m$.}
    \end{array}
    \right.
\end{equation}
In particular, we assume that each sensor can be scheduled to at most one channel and that each channel is assigned to one sensor \cite{pang2022drl}. Then, the constraints on $a_{n,t}$ are given as
    \begin{equation}\label{eq: action constraint}
        \sum_{m=1}^{M} \boldsymbol{\mathbbm{1}}\left( a_{n,t} = m \right) \leq 1, \quad 
        \sum_{n=1}^{N} \boldsymbol{\mathbbm{1}}\left( a_{n,t} = m \right) = 1, 
    \end{equation} where $\boldsymbol{\mathbbm{1}} (\cdot)$ is the indicator function.
    
Considering schedule actions and packet dropouts, sensor $n$'s estimate may not be received by the remote estimator in every time slot. We define the packet reception indicator as
\begin{equation*}
\eta_{n, t} = \left\{
\begin{array}{ll}
1, & \text{if sensor $n$'s packet is received at time $t$}\\
0, & \text{otherwise}.
\end{array}
\right.
\end{equation*} 
Considering the randomness of $\eta_{n,t}$ and assuming that the remote estimator performs state estimation at the beginning of each time slot, the remote state estimate that minimizes the estimation MSE follows the stochastic recursion:
\begin{align}\label{eq:x_hat}
    \hat{\mathbf{x}}_{n, t+1} & = \left\{
                        \begin{array}{ll}
                             \mathbf{A}_n\mathbf{x}_{n, t}^{s}, & {\text{if $\eta_{n,t} = 1$}}\\
                             \mathbf{A}_n \hat{\mathbf{x}}_{n, t}, & {\text{otherwise}},
                        \end{array}
                    \right.
\end{align}
where $\mathbf{A}_n$ is the system matrix of process $n$ defined in \eqref{eq:LTI}.
If sensor $n$'s packet is not received, then the remote estimator propagates its estimate in the previous time slot to estimate the current state.
From \eqref{eq:LTI} and \eqref{eq:x_hat}, we derive the estimation error covariance as
\begin{align}
  \mathbf{P}_{n, t} & \triangleq \mathbb{E} \left[ \left(\hat{\mathbf{x}}_{n, t} - \mathbf{x}_{n, t}\right) \left(\hat{\mathbf{x}}_{n, t} - \mathbf{x}_{n, t} \right)^{\top} \right] \\
& = \left \{
\begin{array}{ll}
\mathbf{A}_n\bar{\mathbf{P}}_{n}\mathbf{A}^\top_n + \mathbf{W}_n, & {\text{if $\eta_{n,t} = 1$}}\\
\mathbf{A}_n{\mathbf{P}}_{n,t}\mathbf{A}^\top_n + \mathbf{W}_n, & {\text{otherwise}},
\end{array} \label{eq:remote MSE}
\right. 
\end{align}
where $\bar{\mathbf{P}}_{n}$ is the local estimation error covariance of sensor $n$ defined under \eqref{eq:local_P}.

Let $\tau_{n,t}\in \{1,2,\dots\}$ denote the age-of-information (AoI) of sensor $n$ at time $t$, which measures the
amount of  time elapsed since the latest sensor packet was successfully received. Then, we have
\begin{equation}\label{eq:tau}
\tau_{n,t+1} =\begin{cases}
1 & {\text{if $\eta_{n,t} = 1$}}\\
\tau_{n,t}+1 & {\text{otherwise.}}
\end{cases}
\end{equation}
If a sensor is frequently scheduled at good channels, then the corresponding average AoI is small. However, due to the scheduling constraint~\eqref{eq: action constraint}, this is often not possible. This leads to estimator stability issues, as studied, e.g., in~\cite{pang2022drl}.

From \eqref{eq:remote MSE} and \eqref{eq:tau}, the error covariances can be written in terms of the AoI as
\begin{equation}\label{eq:MSE}
    \mathbf{P}_{n,t} = f_{n}^{\tau _{n,t}}(\bar{\mathbf{P}}_{n}),
\end{equation}
where
$f_{n}(\mathbf{X}) = \mathbf{A}_{n} \mathbf{X} \mathbf{A}_{n}^{\top} + \mathbf{W}_{n}$ and $f^{\tau+1}_{n}(\cdot) =f_n(f^{\tau}_{n}(\cdot))$. It has been proved that the estimation MSE, i.e., $\operatorname{Tr}(\mathbf{P}_{n,t})$, \emph{monotonically increases with the AoI state $\tau_{n,t}$~\cite{liu2021remoteMF}}.

\section{Problem Formulation and Threshold structure} \label{sec: problem}
In this paper, we aim to find a dynamic scheduling policy $\pi(\cdot)$ that uses the AoI states of all sensors, as well as the channel states to minimize the expected total discounted estimation MSE of all $N$ processes over the infinite time horizon.
\begin{problem}\label{pro1}
\begin{equation}
    \max_{\pi} \lim_{T \to \infty} \mathbb{E} \left[ \sum_{t=1}^{T} \sum_{n=1}^N -\gamma^t \operatorname{Tr}(\mathbf{P}_{n,t}) \right],
\end{equation}
where $\gamma\in(0,1)$ is a discount factor.
\end{problem}

Problem~\ref{pro1} is a Markovian sequential decision-making problem. This is because the instantaneous estimation MSE, $\mathbf{P}_{n,t}$, only depends on the AoI state $\tau_{n,t}$ in \eqref{eq:MSE}, which is Markovian~\eqref{eq:tau}, and the channel states are i.i.d..
Therefore, we formulate Problem~\ref{pro1} as an MDP.\footnote{Not all MDPs for transmission scheduling of remote estimation systems have a feasible solution. However, once the remote estimation stability condition in terms of the dynamic process parameters and the channel statistics is satisfied, the MDP has a solution. We assume that the remote estimation stability condition of our system is satisfied, and only focus on the optimal solution of the MDP in the rest of the paper. The detailed stability condition can be found in our previous work~\cite{pang2022drl}.}

\subsection{MDP Formulation}\label{sec:MDP}
1) The state of the MDP is defined as $\mathbf{s}_{t} \triangleq \left( \bm{\tau}_{t}, \mathbf{H}_{t} \right) \in \mathcal{S} \triangleq \mathbb{N}^N \times \mathcal{H}^{N \times M}$, where $\bm{\tau}_{t} = (\tau_{1,t}, \tau_{2,t}, \dots, \tau_{N,t}) \in \mathbb{N}^N$ is the AoI state vector. Thus, $\mathbf{s}_{t}$ takes into account both the AoI and channel states.
    
2) The overall schedule action of the $N$ sensors is defined as $\mathbf{a}_{t} = (a_{1,t}, a_{2,t}, \dots, a_{N,t}) \in \mathcal{A} \triangleq \left\{ 0, 1, 2, \dots, M \right\}^N$ under the constraint~\eqref{eq: action constraint}. There are $N!/(N-M)!$ actions of the $N$-choose-$M$ problem in total.
The policy $\pi$ is a state-action mapping, i.e., $\mathbf{a}_{t} = \pi(\mathbf{s}_{t}) $.
    
3) The transition probability $\operatorname{Pr}(\mathbf{s}_{t+1}|\mathbf{s}_{t}, \mathbf{a}_{t})$ is the probability of the next state $\mathbf{s}_{t+1}$ given the current state $\mathbf{s}_{t}$ and the  action $\mathbf{a}_{t}$. 
Since the state transition is independent of the time index given the action $\mathbf{a}$ and the state $\mathbf{s}$, we drop the subscript $t$ here and use $\mathbf{s}$ and $\mathbf{s}^+$ to represent the current and the next states, respectively.
Due to the i.i.d. fading channel states, we have 
\begin{equation}\label{eq: transition p}
    \operatorname{Pr}(\mathbf{s}^+|\mathbf{s}, \mathbf{a}) = \operatorname{Pr}(\bm{\tau}^+|\bm{\tau},\mathbf{H}, \mathbf{a}) \operatorname{Pr}(\mathbf{H}^+),
\end{equation}
where $\operatorname{Pr}(\mathbf{H}^+)$ can be obtained from \eqref{eq:q}, and  $\operatorname{Pr}(\bm{\tau}^+|\bm{\tau},\mathbf{H}, \mathbf{a}) = \prod_{n=1}^N \operatorname{Pr}(\tau^+_{n}|\tau_{n}, \mathbf{H}, a_{n})$ and
\begin{align}
    & \operatorname{Pr}(\tau^+_{n}|\tau_{n}, \mathbf{H}, a_{n})  = \left\{
    \begin{array}{l}
        p_{n, m}, \qquad \; \; \text{if $\tau^+_{n} = 1, a_{n} = m$} \\
        1-p_{n, m}, \ \ \, \text{if $\tau^+_{n} = \tau_{n} + 1, a_{n} = m$}\\
        1, \qquad \quad \,  \text{\ \ if $\tau^+_{n} = \tau_{n} + 1, a_{n} = 0$} \\
        0, \qquad \quad \,  {\ \ \text{otherwise.}}
    \end{array}
    \right.
\end{align}
which is derived based on~\eqref{eq:q} and~\eqref{eq:tau}.

4) The immediate reward of Problem~\ref{pro1} at time $t$ is defined as the negative sum estimation MSE, $\sum_{n=1}^N -\operatorname{Tr} (\mathbf{P}_{n,t})$. Since $\mathbf{P}_{n,t}$ defined in \eqref{eq:MSE} is a function of $\tau_{n,t}$, the reward is represented as $r(\mathbf{s}_{t})$ and \emph{monotonically decreases with each AoI state}.

\subsection{Threshold Structure of the Optimal MDP Solution}
 Unlike the existing works~\cite{wu2018optimalmulti, wu2020optimalmulti}, which only considered oversimplified systems over static channels, we aim to derive the structural property of the optimal policy of the general multi-sensor-multi-channel system over fading channels as below.

\begin{definition}[Channel-State Threshold Policy]\label{def: threshold-channel}
	For a channel-state threshold scheduling policy, if channel~$m$ is assigned to sensor~$n$ at the state $\mathbf{s} = (\bm{\tau}, \mathbf{H})$, then for state $\mathbf{s}' = (\bm{\tau}, \mathbf{H}'_{n,m})$, where $\mathbf{H}'_{n,m}$ is identical to $\mathbf{H}$ except the sensor-$n$-channel-$m$ state with $h'_{n,m} > h_{n,m}$, then channel $m$ is still assigned to sensor~$n$.
\end{definition}

\begin{definition}[AoI-State Threshold Policy]\label{def: threshold-AoI}
	For an AoI-state threshold scheduling policy, if channel~$m$ is assigned to sensor~$n$ at the state $\mathbf{s} = (\bm{\tau}, \mathbf{H})$, then for state $\mathbf{s}' = (\bm{\tau}'_{(n)}, \mathbf{H})$, where $ \bm{\tau}'_{(n)}$ is idential to $\bm{\tau}$ except sensor $n$'s AoI with $\tau'_n \geq \tau_{n}$, then either channel $m$ or a better channel is assigned to sensor~$n$.
\end{definition}

Definition~\ref{def: threshold-channel} states that sensor $n$ is scheduled at channel $m$ at a certain state, if the channel quality of $h_{n,m}$ improves while the AoI and the other channel states are the same, then the threshold policy still assigns channel $m$ to sensor $n$. 
For Definition~\ref{def: threshold-AoI}, if the AoI state of sensor $n$ is increased while the other states remain the same, the  policy must schedule sensor $n$ to a channel that is no worse than the previous one. 

To illustrate that an optimal policy may have a threshold structure, we find the optimal policy of a two-sensor-single-channel system by solving the MDP with the conventional value iteration algorithm, which is illustrated in Fig.~\ref{fig: threshold}. 
We see that action-switching curves exist in both AoI and channel state spaces, and the properties in Definitions~\ref{def: threshold-channel} and~\ref{def: threshold-AoI} are observed.
Inspired by the above result, in the following, we will prove that the optimal policy has the structural properties in Definitions~\ref{def: threshold-channel} and~\ref{def: threshold-AoI}.\footnote{We note that in \cite{wu2018optimalmulti, wu2020optimalmulti}, only threshold structural results in terms of the AoI states are proved over static channels. Moreover, the proofs have some limitations: 
1) In \cite{wu2018optimalmulti}, the authors proved the threshold structure for a multi-dimensional state and action scenario based on Theorem 8.11.3 in \cite{puterman2014markov}, which, however, can only be used for single-dimensional scenarios.
2) In \cite{wu2020optimalmulti}, the proof of the optimal policy's structural property only considered a part of the action space; however, the full action space must be examined to show the optimally.}
 
\begin{figure}[t]
    \centering
    \begin{subfigure}[Schedule actions at AoI states]
        {
        \begin{minipage}[t]{0.45\linewidth}
        \centering
        \includegraphics[width=\textwidth]{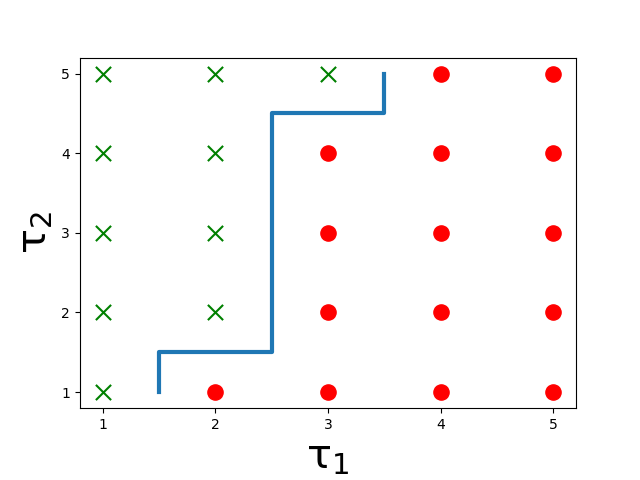}
        \end{minipage}
        \label{fig: threshold AoI}
        }
    \end{subfigure}
    \hspace{-0.9cm}
    \begin{subfigure}[Schedule actions at channel states]
        {
        \begin{minipage}[t]{0.45\linewidth}
        \centering
        \includegraphics[width=1\textwidth]{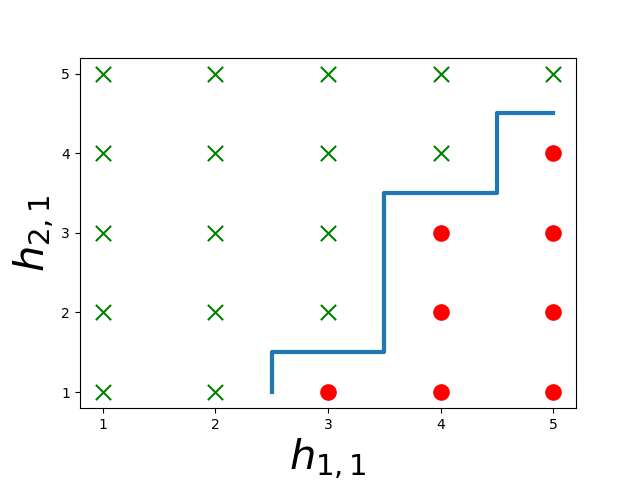}
        \end{minipage}
        \label{fig: threshold channel}
        }
    \end{subfigure}
    \vspace{-0.3cm}
    \caption{Structure of the optimal scheduling policy with $N=2$ and $M=1$, where $\bullet$ and $\times$ represent the schedule of sensor 1 and 2, respectively.}
    \vspace{-0.7cm}
    \label{fig: threshold}
\end{figure}

\section{Threshold Structure of Optimal Policies} \label{sec: proof of threshold structure}
We derive the structural properties of the optimal scheduling policy by using value iteration concepts, which require the definition of the optimal value function, $V^{*}(\mathbf{s}_{t}): \mathcal{S} \to \mathbb{R}$, and the state-action value function, $Q(\mathbf{s}_{t}, \mathbf{a}_{t}): \mathcal{S} \times \mathcal{A} \to \mathbb{R}$ as below. 

Given the current state $\mathbf{s}_{t} = (\bm{\tau}_{t}, \mathbf{H}_{t})$, the optimal value function is the maximum expected discounted sum of the future cost, i.e., achieved by the optimal policy $\pi^{*}(\cdot)$.
\begin{equation}\label{eq: optimal value function}
    V^{*}(\mathbf{s}_{t}) = \max_{\pi} \mathbb{E}\left[ \sum_{t'=t}^{\infty} \gamma^{t'-t} r\left(\mathbf{s}_{t'}\right)\right].
\end{equation}
The optimal value function satisfies the Bellman equation:
\begin{align}
    & V^{*}(\mathbf{s}_{t})  = r(\mathbf{s}_{t}) + \gamma \max_{\mathbf{a}_{t} \in \mathcal{A}}  \left[  \sum_{\mathbf{s}_{t+1}} \operatorname{Pr}(\mathbf{s}_{t+1}|\mathbf{s}_{t}, \mathbf{a}_{t})V^{*}(\mathbf{s}_{t+1})\right], \label{eq:V}
\end{align}
where the optimal action $\mathbf{a}^{*}$ is obtained by the optimal policy $\pi^{*}(\cdot)$, i.e.,
\begin{equation}\label{eq:optimal V action}
    \mathbf{a}_{t}^{*} 
    \triangleq \pi^{*}(\mathbf{s}_{t}) 
    = \mathop{\arg\max}_{\mathbf{a}_{t}\in \mathcal{A}}  \left[ \sum_{\mathbf{s}_{t+1}} \operatorname{Pr}(\mathbf{s}_{t+1}|\mathbf{s}_{t}, \mathbf{a}_{t})V^{*}(\mathbf{s}_{t+1})\right].
\end{equation}
Given the current state-action pair $\mathbf{s}_{t}$ and $\mathbf{a}_{t}$, the state-action value function, which is also called the Q-value function, measures the expected discounted sum of the future cost under the optimal policy $\pi^{*}(\cdot)$ as
\begin{equation}\label{eq:Bellman_Q}
    Q(\mathbf{s}_{t}, \mathbf{a}_{t}) 
    = r(\mathbf{s}_{t}) + \gamma \sum_{\mathbf{s}_{t+1}}\operatorname{Pr}(\mathbf{s}_{t+1}|\mathbf{s}_{t}, \mathbf{a}_{t}) V^{*}(\mathbf{s}_{t+1}). 
\end{equation}
From~\eqref{eq:V} and~\eqref{eq:Bellman_Q}, it directly follows that the optimal value function and the Q-value function satisfy:
\begin{equation}\label{ineq:V&Q}
    V^{*}(\mathbf{s}_{t}) = Q(\mathbf{s}_{t}, \mathbf{a}_{t}^*) \geq Q(\mathbf{s}_{t}, \mathbf{a}_{t}).
\end{equation}
For notation simplicity, we use $\mathbf{a}, \mathbf{s}, \mathbf{s}^{+}$ to represent $\mathbf{a}_{t}, \mathbf{s}_{t}, \mathbf{s}_{t+1}$ in the following.

We prove the channel-state and the AoI-state threshold properties of the optimal scheduling policy based on the classical value iteration algorithm, as it can achieve the optimal solution~\cite{puterman1990markov, hernandez2012further}.
In the value iteration, the initial value function and its  $\tilde{t}$-th iteration are $V^{0}(\mathbf{s})\in \mathcal{V} $ and $V^{\tilde{t}}(\mathbf{s}) \in \mathcal{V}$, respectively, where $\mathcal{V}$ is the set of any measurable function, i.e., $\mathcal{V}: \mathcal{S} \to \mathbb{R}$. At the $\tilde{t}$-th iteration, we have $V^{\tilde{t}+1} = \mathsf{B} \left[V^{\tilde{t}}\right]$, where $\mathsf{B}[\cdot]: \mathcal{V} \to \mathcal{V}$ is the Bellman operator:
\begin{align}
    \mathsf{B} \left[V^{\tilde{t}}\right] (\mathbf{s}) = r(\mathbf{s}) + \gamma \max_{\mathbf{a} \in \mathcal{A}} \left[ \sum_{\mathbf{s}^{+}} \operatorname{Pr}(\mathbf{s}^{+}|\mathbf{s}, \mathbf{a}) V^{\tilde{t}} (\mathbf{s}^{+})\right]. \label{eq: bellman operator}
\end{align}
We next elucidate the convergence and the optimality of value iterations.
\begin{lemma}[\!\cite{puterman1990markov, hernandez2012further}\,]\label{lemma:converge of V*}
    If the optimal policy exists, then the operator $\mathsf{B}$ has a unique fixed point $V^{*} \in \mathcal{V}$ and for all $V^{0} \in \mathcal{V}$, the sequence $\{V^{\tilde{t}}\}$ defined by $V^{\tilde{t}+1} = \mathsf{B} [V^{\tilde{t}}]$ converges in norm to $V^{*}$, i.e.
    \begin{equation}
        \lim_{\tilde{t} \to \infty} V^{\tilde{t}} = V^{*}.
    \end{equation}
\end{lemma}

Before proceeding further, we derive the following technical lemma about the monotonicity of the optimal value function. 
\begin{lemma}[Monotonicity]\label{lemma:monotone}
    Consider states $\mathbf{s} = (\bm{\tau}, \mathbf{H})$ and $\mathbf{s}' = (\bm{\tau}'_{(i)}, \mathbf{H})$, where $\bm{\tau}'_{(i)} = (\tau_{1}, \dots, \tau'_{i}, \dots, \tau_{N})$ and $\tau'_{i} \geq \tau_{i}$. The following holds 
    \begin{equation}
        V^{*}(\mathbf{s}') \leq V^{*}(\mathbf{s}).
    \end{equation}
\end{lemma}
\begin{proof}\label{proof:lemma1}
    See Appendix~\ref{proof: monotonicity}.
\end{proof}
Lemma~\ref{lemma:monotone} will assist in the comparison of value functions with different AoI states in the proof of the channel-state threshold and the AoI-state threshold properties.

\subsection{Channel-State Threshold Property}
In this part, we completely prove that the optimal policy of a general multi-sensor-multi-channel system has the channel-state threshold property in Definition~\ref{def: threshold-channel}.

\begin{theorem}\label{theo: multiple channel}
    The optimal policy $\pi^{*}(\cdot)$ of the multi-sensor-multi-channel system has the channel-state threshold property in Definition~\ref{def: threshold-channel}.
\end{theorem}
\begin{proof}
 From the Q-value definition in \eqref{eq:Bellman_Q},  the theorem can be translated as: if for state $\mathbf{s} = (\bm{\tau}, \mathbf{H})$, the inequality $Q(\mathbf{s}, \mathbf{a}^{*}) \geq Q(\mathbf{s}, \mathbf{a}),\forall \mathbf{a}\in\mathcal{A}$ exists, where ${a}^{*}_{i} = m$, then for state $\mathbf{s}' = (\bm{\tau}, \mathbf{H}'_{i,m})$, where $\mathbf{H}'_{i,m}$ and $\mathbf{H}$ are identical except the element $h'_{i,m} > h_{i,m}$, the following inequality holds $Q(\mathbf{s}', {\mathbf{a}'}^{*}) \geq Q(\mathbf{s}', \mathbf{a})$, where ${\mathbf{a}'}^{*}$ is the optimal action at the state $\mathbf{s}'$ with ${a'_{i}}^{*} = m$. Since $Q(\mathbf{s}', {\mathbf{a}'}^{*}) \geq Q(\mathbf{s}', \mathbf{a}^{*})$, we only need to prove $Q(\mathbf{s}', \mathbf{a}^{*})  \geq Q(\mathbf{s}', \mathbf{a}')$, where ${a}'_{i} \neq m$.
To prove $Q(\mathbf{s}', \mathbf{a}^{*}) \geq Q(\mathbf{s}', \mathbf{a}')$, we will prove $Q(\mathbf{s}', \mathbf{a}^{*}) \geq Q(\mathbf{s}, \mathbf{a}^{*})$ and then $Q(\mathbf{s}, \mathbf{a}') = Q(\mathbf{s}', \mathbf{a}')$. In the following, we use $\mathbf{H}'$ to represent $\mathbf{H}'_{i,m}$ for the notation simplicity.

First, from~\eqref{eq:q} and~\eqref{eq:tau}, we have
    \begin{align}
        \operatorname{Pr}(\bm{\tau}^{+}|\bm{\tau}, \mathbf{H}, \mathbf{a}) 
        = \prod_{n=1}^{N} \operatorname{Pr} (\tau_{n}^{+}|\tau_{n}, \mathbf{h}_{n}, a_{n}) 
        = \operatorname{Pr} (\tau_{i}^{+}|\tau_{i}, \mathbf{h}_{i}, a_{i}) \operatorname{Pr} (\bm{\tau}_{\backslash i}^{+}|\bm{\tau}_{\backslash i}, \mathbf{H}_{\backslash i}, \mathbf{a}_{\backslash i}), \label{eq: transition tau}
    \end{align}
    where $\bm{\tau}_{\backslash i} = (\tau_{1}, \dots, \tau_{i-1}, \tau_{i+1}, \dots, \tau_{N})$ and $\mathbf{a}_{\backslash i} = (a_{1}, \dots, a_{i-1}, a_{i+1}, \dots, a_{N})$ represent the AoI states and the actions of all sensors except sensor $i$, respectively, and  $\mathbf{h}_{i} = (h_{i,1}, h_{i,2}, \dots, h_{i,M})$ is the vector channel states between sensor $i$ and the remote estimator and $\mathbf{H}_{\backslash i} = (\mathbf{h}_{1}, \dots, \mathbf{h}_{i-1}, \\ \mathbf{h}_{i+1}, \dots,  \mathbf{h}_{N})$.
    By using~\eqref{eq: transition p} and~\eqref{eq: transition tau}, it can be derived that \vspace{-0.5cm}
    \begin{align}
        Q(\mathbf{s}, \mathbf{a}) 
        & = r(\mathbf{s}) + \gamma \sum_{\mathbf{s}^{+}} \operatorname{Pr}(\mathbf{s}^{+}|\mathbf{s}, \mathbf{H}, \mathbf{a}) V^{*}(\mathbf{s}^{+}) \\
        & = r(\mathbf{s}) + \gamma \sum_{\mathbf{H}^{+}} \sum_{\bm{\tau}^{+}} \operatorname{Pr} (\mathbf{H}^{+}) \operatorname{Pr} (\bm{\tau}^{+}|\bm{\tau}, \mathbf{H}, \mathbf{a}) V^{*}(\mathbf{s}^{+})\\
        & = r(\mathbf{s}) + \gamma \sum_{\mathbf{H}^{+}} \sum_{\bm{\tau}_{\backslash i}^{+}} \sum_{\tau_{i}^{+}} \operatorname{Pr} (\mathbf{H}^{+}) \operatorname{Pr} (\bm{\tau}_{\backslash i}^{+}|\bm{\tau}_{\backslash i}, \mathbf{H}_{\backslash i}, \mathbf{a}_{\backslash i})  \operatorname{Pr} (\tau_{i}^{+}|\tau_{i}, \mathbf{h}_{i}, a_{i}) V^{*}(\mathbf{s}^{+}). \label{eq: Q i -i}
    \end{align}
Based on~\eqref{eq: Q i -i}, we have
\begin{align}
    & \hspace{-0.2cm} Q(\mathbf{s}', \mathbf{a}^{*}) = r(\mathbf{s}') \!+\! \gamma \sum_{{\mathbf{H}'}^{+}} \sum_{\bm{\tau}_{\backslash i}^{+}}\sum_{\tau_{i}^{+}} \operatorname{Pr} ({\mathbf{H}'}^{+}) \operatorname{Pr} (\bm{\tau}_{\backslash i}^{+}|\bm{\tau}_{\backslash i}, \mathbf{H}_{\backslash i}, \mathbf{a}^{*}_{\backslash i}) \operatorname{Pr} (\tau_{i}^{+}|\tau_{i}, \mathbf{h}'_{i}, a^{*}_{i})  V^{*}({\mathbf{s}'}^{+}) \\
    & \hspace{-0.2cm} \quad \ \!\geq\! r(\mathbf{s}) \!+\! \gamma \sum_{\mathbf{H}^{+}} \sum_{\bm{\tau}_{\backslash i}^{+}}  \sum_{\tau_{i}^{+}} \operatorname{Pr} (\mathbf{H}^{+}) \operatorname{Pr} (\bm{\tau}_{\backslash i}^{+}|\bm{\tau}_{\backslash i}, \!\mathbf{H}_{\backslash i}, \!\mathbf{a}^{*}_{\backslash i})  \operatorname{Pr} (\tau_{i}^{+}|\tau_{i}, \!\mathbf{h}_{i},\! a^{*}_{i})  V^{*}(\mathbf{s}^{+}) \!=\! Q(\mathbf{s}, \mathbf{a}^{*}), \label{eq: Q(s',a*) > Q(s,a*)}
\end{align}
where the inequality is derived by replacing the parameter ${\mathbf{H}'}^{+}$ with ${\mathbf{H}}^{+}$, and then using $r(\mathbf{s}') = r(\mathbf{s})$, $h'_{i,m} > h_{i,m}, a^{*}_{i} = m$, and the following inequality
\begin{align}
    p'_{i,m} \!V^{*}\!(1, \! \bm{\tau}_{\backslash i}^{+}, \! \mathbf{H}^{+}) \!+\! (1\!\!-\!p'_{i,m}) V^{*}\!(\tau_{i}\!+\!1, \!\bm{\tau}_{\backslash i}^{+}, \!\mathbf{H}^{+})
    \!\geq\! p_{i,m} \!V^{*}\!(1, \!\bm{\tau}_{\backslash i}^{+}, \!{\mathbf{H}}^{+}) \!+\! (1\!\!-\!p_{i,m}) V^{*}\!(\tau_{i}\!+\!1, \!\bm{\tau}_{\backslash i}^{+}, \!{\mathbf{H}}^{+}) \hspace{-0.5cm} \label{eq: p'V > pV}
\end{align}
achieved by $p'_{i,m} \geq p_{i,m}$ and $V^{*}(1, \! \bm{\tau}_{\backslash i}^{+}, \! \mathbf{H}^{+})\geq V^{*}(\tau_{i}\!+\!1, \!\bm{\tau}_{\backslash i}^{+}, \!\mathbf{H}^{+})$ from Lemma~\ref{lemma:monotone}.

Second, we derive that
\begin{align}
    & \hspace{-0.2cm} Q(\mathbf{s}, \mathbf{a}) = r(\mathbf{s}) \!+\! \gamma \sum_{\mathbf{H}^{+}} \sum_{\bm{\tau}_{\backslash i}^{+}} \sum_{\tau_{i}^{+}} \operatorname{Pr} (\mathbf{H}^{+})   \operatorname{Pr} (\bm{\tau}_{\backslash i}^{+}|\bm{\tau}_{\backslash i}, \mathbf{H}_{\backslash i}, \mathbf{a}'_{\backslash i})  \operatorname{Pr} (\tau_{i}^{+}|\tau_{i}, \mathbf{h}_{i}, a'_{i}) V^{*}(\mathbf{s}^{+}) \\
    & \hspace{-0.2cm} \quad \ \! =\! r(\mathbf{s}') \!+\! \gamma \sum_{{\mathbf{H}'}^{+}} \sum_{\bm{\tau}_{\backslash i}^{+}} \sum_{\tau_{i}^{+}} \operatorname{Pr} ({\mathbf{H}'}^{+}) \operatorname{Pr} (\bm{\tau}_{\backslash i}^{+}|\bm{\tau}_{\backslash i},\! \mathbf{H}_{\backslash i},\! \mathbf{a}'_{\backslash i})\! \operatorname{Pr} (\tau_{i}^{+}|\tau_{i}, \!\mathbf{h}'_{i},\! a'_{i}) V^{*}({\mathbf{s}'}^{+}) \!=\! Q(\mathbf{s}', \!\mathbf{a}), \label{eq: Q(s,a) = Q(s',a)}
\end{align}
where the second equality is derived by replacing the parameter ${\mathbf{H}'}^{+}$ with ${\mathbf{H}}^{+}$, and then using $h'_{i,m} > h_{i,m}, a'_{i} \neq m$, and $\operatorname{Pr} (\tau_{i}^{+}|\tau_{i}, \mathbf{h}_{i}, a'_{i}) V^{*}(\mathbf{s}^{+}) = \operatorname{Pr} (\tau_{i}^{+}|\tau_{i}, \mathbf{h}'_{i}, a'_{i}) V^{*}({\mathbf{s}'}^{+})$.
\end{proof}

\subsection{AoI-State Threshold Property}\label{sec: property (ii)}
We consider three network scenarios with different numbers of sensors and channels as below.

\subsubsection{Two-sensor-single-channel systems}
As presented in Theorem~\ref{theo:2-s-1-c}, we prove that the optimal policy of a two-sensor-single-channel system is an AoI-state threshold policy in Definition~\ref{def: threshold-AoI}.
In other words, if the optimal action at a state is to schedule sensor $n$, then the optimal action is still to schedule sensor $n$ at all states that only increase sensor $n$'s AoI.
\begin{theorem} \label{theo:2-s-1-c}
    The optimal policy $\pi^{*}(\cdot)$ of the two-sensor-single-channel system has the AoI-state threshold property in Definition~\ref{def: threshold-AoI}: suppose that for state $\mathbf{s} = (\bm{\tau}, \mathbf{H})$, the optimal action is to schedule sensor $i$, i.e. $a^*_{i} = 1$, then for state $\mathbf{s}' = (\bm{\tau}'_{(i)}, \mathbf{H}),$ where $ \tau'_{i} \geq \tau_{i}$, the optimal action is still $a'^*_{i} = 1$.
\end{theorem}

\begin{remark}[Analytical Challenges]\label{remark:difficulty}
Although Theorem~\ref{theo:2-s-1-c} looks simple, the proof is highly nontrivial.

First, the main difficulty of the proof is that different actions cause the transition of the AoI states in multiple dimensions, resulting in the comparison of multi-dimensional value functions.
However, Lemma~\ref{lemma:monotone} can only compare the optimal value functions with AoI state changes within one dimension.
For example, given states $\mathbf{s} = (\bm{\tau}, \mathbf{H}), \mathbf{s}' = (\bm{\tau}', \mathbf{H})$, and $\mathbf{s}^{\circ} = (\bm{\tau}^{\circ}, \mathbf{H})$, where $\bm{\tau} = (\tau_{i}, \tau_{j}), \bm{\tau}' = (\tau_{i}', \tau_{j}), \bm{\tau}^{\circ} = (\tau_{i}', \tau_{j}''), \tau_{i}' \geq \tau_{i}$, and $\tau_{j}'' \leq \tau_{j}$, it is easy to have $V(\mathbf{s}) \geq V(\mathbf{s}')$, but is not possible to compare $V(\mathbf{s})$ and $V(\mathbf{s}^{\circ})$.

Second, one may think that the proof of Theorem~\ref{theo:2-s-1-c} is intuitive as a larger AoI of sensor~$i$ results in a lower reward, and scheduling sensor $i$ can improve the reward more efficiently than scheduling other sensors. This is a misunderstanding. Theorem~\ref{theo:2-s-1-c} is equivalent to saying that the inequality $Q(\mathbf{s}', \mathbf{a}^*) \geq Q(\mathbf{s}', \mathbf{a})$ can be derived from $Q(\mathbf{s}, \mathbf{a}^*) \geq Q(\mathbf{s}, \mathbf{a})$, where $a^*_{i}=1$.
We drop the constant channel state, $\mathbf{H}$, in the optimal value functions for simplicity.
Based on~\eqref{eq:Bellman_Q}, $Q(\mathbf{s}', \mathbf{a}^*) \geq Q(\mathbf{s}', \mathbf{a})$ is equal to 
    \begin{align}
        \!\!\!p_{i,1} V^{*}(1, \tau_{j}\!+\!1) \!+\! (1\!-\!p_{i,1}) V^{*}(\tau'_{i}\!+\!1, \tau_{j}\!+\!1)
        \!\geq\! p_{j,1} V^{*}(\tau'_{i}\!+\!1, 1) \!+\! (1\!-\!p_{j,1}) V^{*}(\tau'_{i}\!+\!1, \tau_{j}\!+\!1),\label{eq:Q(s',a*) > Q(s',a)}
    \end{align}
and $Q(\mathbf{s}, \mathbf{a}^*) \geq Q(\mathbf{s}, \mathbf{a})$ is equal to
    \begin{align}
        \!\!\!p_{i,1} V^{*}(1, \tau_{j}\!+\!1) \!+\! (1\!-\!p_{i,1}) V^{*}(\tau_{i}\!+\!1, \tau_{j}\!+\!1)
        \!\geq\! p_{j,1} V^{*}(\tau_{i}\!+\!1, 1) \!+\! (1\!-\!p_{j,1}) V^{*}(\tau_{i}\!+\!1, \tau_{j}\!+\!1).\label{eq:Q(s,a*) > Q(s,a)}
    \end{align}
We see that~\eqref{eq:Q(s',a*) > Q(s',a)} cannot be derived directly based on~\eqref{eq:Q(s,a*) > Q(s,a)}, since it also related to the packet success rates of different sensors, i.e. $p_{i,1}$ and $p_{j,1}$.
\end{remark}

\begin{proof}
As mentioned in Remark~\ref{remark:difficulty}, we will prove \eqref{eq:Q(s',a*) > Q(s',a)} based on \eqref{eq:Q(s,a*) > Q(s,a)}, which depends on the key technical lemma below in terms of the optimal value functions and the packet success rates, which depends on the channel states.
\begin{definition}[Meet and Joint~\cite{topkis1998supermodularity, puterman2014markov}]
    Let $x' \vee x'' = \max \{ x', x'' \}$ and $x' \wedge x'' = \min \{ x', x'' \}$ denote the join and meet of two real numbers, respectively. Define $\mathbf{x}' \vee \mathbf{x}'' = (x'_1 \vee x''_1, \dots, x'_n \vee x''_n)$ and $\mathbf{x}' \wedge \mathbf{x}'' = (x'_1 \wedge x''_1, \dots, x'_n \wedge x''_n)$ as the joint and meet of the vectors $\mathbf{x}'$ and $\mathbf{x}'' \in \mathbb{R}^n$, respectively.
\end{definition}
\begin{lemma}[Probabilistic Supermodularity, $N=2$, $M=1$]\label{lemma:probabilistic sup}
    Given states $\mathbf{s} = (\bm{\tau}, \mathbf{H})$ and $ \mathbf{s}^\circ = (\bm{\tau}^\circ, \mathbf{H})$, where $\bm{\tau}^\circ = (\tau''_i, \tau'_{j}), \tau''_{i} \leq \tau_{i}$, $ \tau'_{j} \geq \tau_{j}$, and $i\neq j \in \{1,2\}$, the following holds
    \begin{equation}\label{eq:sup1}
        p_{j,,1} V^{*}(\mathbf{s} \wedge \mathbf{s}^\circ) + (p_{j,1} - p_{i,1}) V^{*}(\mathbf{s} \vee \mathbf{s}^\circ) \geq p_{j,1} V^{*}(\mathbf{s}) + (p_{j,1} - p_{i,1}) V^{*}(\mathbf{s}^\circ).
    \end{equation}
\end{lemma}

\begin{proof}
    See Appendix~\ref{proof: sup and sub, N=2}.
\end{proof}

In what follows, we denote the state as $\mathbf{s} = (\tau_{i}, \tau_{j})$, since the transition of the channel states is constant during the proof. From \eqref{eq:Q(s,a*) > Q(s,a)}, it directly follows that 
\begin{align}
p_{i,1} V^{*}(1, \tau_{j} \!+\! 1) 
    \geq p_{j,1} V^{*}(\tau_{i} \!+\! 1, 1) \!-\! (p_{j,1} \!-\! p_{i,1}) V^{*}(\tau_{i} \!+\! 1, \tau_{j} \!+\! 1). \label{eq:Q<Q_1}
\end{align}
By applying Lemma~\ref{lemma:probabilistic sup} to the right-hand side of \eqref{eq:Q<Q_1}, we have
\begin{align*}
    & p_{i,1} V^{*}(1, \tau_{j}+1) \geq p_{j,1} V^{*}(\tau'_{i,1} + 1, 1)  -  (p_{j,1} - p_{i,1}) V^{*}(\tau'_{i} + 1, \tau_{j} + 1),
\end{align*}
which is exactly \eqref{eq:Q(s',a*) > Q(s',a)}. Thus, we have proved $Q(\mathbf{s}', \mathbf{a}^*) \geq Q(\mathbf{s}', \mathbf{a})$.
\end{proof}

\subsubsection{Multi-sensor-single-channel systems}
It is intractable to prove that the optimal policy has the AoI-state threshold property in this case. 
Similar to the two-sensor case, we first try to prove the probabilistic supermodularity. However, there are more than two AoI states changing when applying value iterations, making it difficult to find a set of useful inequalities to prove the target inequality.
We can also construct some other inequalities in terms of the optimal value function and the packet success rates that can be used to prove the threshold property, but these inequalities cannot be proved either.

Although the AoI-state threshold property of the optimal policy is not fully derived, we prove an asymptotic structural property as below.
\begin{theorem}\label{theo:N-s-1-c}
   The optimal policy $\pi^{*}(\cdot)$ of a multi-sensor-single-channel system with $N >2$ has an asymptotic AoI-state threshold property: suppose that for state $\mathbf{s} = (\bm{\tau}, \mathbf{H})$, the optimal action is to transmit the estimation of sensor $i$, i.e. $ a^*_{i} = 1$, then for all states $\mathbf{s}' = (\bm{\tau}'_{(i)}, \mathbf{H}),$ where $ \tau'_{i} \gg \tau_{i} $, the optimal policy is still $a'^*_{i} = 1$.
\end{theorem} 
Theorem~\ref{theo:N-s-1-c} shows that the threshold structure still exists in the optimal policy of a multi-sensor-single-channel system, at least in the large AoI domain.

\begin{proof}
This theorem is equivalent to stating that if for state $\mathbf{s} = (\bm{\tau}, \mathbf{H})$, the inequality $Q(\mathbf{s}, \mathbf{a}^*) \geq Q(\mathbf{s}, \mathbf{a}), \forall \mathbf{a} \in \mathcal{A}$ exists, where $a_i^* = 1$, then for state $\mathbf{s}' = (\bm{\tau}'_{(i)}, \mathbf{H}),$ where $ \tau'_{i} \gg \tau_{i}$, we have $Q(\mathbf{s}', \mathbf{a}^*) \geq Q(\mathbf{s}', \mathbf{a})$. Similar to Theorem~\ref{theo:2-s-1-c}, we need to develop the following lemmas yet in an asymptotic manner.

\begin{lemma}\label{lemma:much larger}
    Given states $\mathbf{s} = (\bm{\tau}, \mathbf{H})$ and $\mathbf{s}' = (\bm{\tau}'_{(i)}, \mathbf{H}) $, where $\tau'_{i} \gg \tau_{i}$, the following holds
    \begin{equation}
        V^{*}(\mathbf{s}') \ll V^{*}(\mathbf{s}).
    \end{equation}
\end{lemma}
\begin{proof}
    See Appendix~\ref{proof: much larger}.
\end{proof}
Based on Lemma~\ref{lemma:much larger}, we derive the asymptotic probabilistic supermodularity below.
\begin{lemma}[Asymptotic Probabilistic Supermodularity, $N>2$, $M=1$]\label{lemma: probabilistic sup, N>2}
    Given states $\mathbf{s} = (\bm{\tau}, \mathbf{H})$ and $ \mathbf{s}^{\circ} = (\bm{\tau}^{\circ}, \mathbf{H})$, where $\bm{\tau} = (\tau_1, \dots, \tau_{i}, \dots, \tau_{j}, \dots, \tau_N)$ and $ \bm{\tau}^{\circ} = (\tau_1, \dots, \tau''_i, \dots, \tau'_{j}, \dots,\tau_N)$ with $\tau_i \gg \tau''_{i}$, $\tau_{j} \leq \tau'_{j}$, and $i\neq j \in \{1,\dots,N\}$, the probabilistic supermodularity in~\eqref{eq:sup1} holds.
\end{lemma}
\begin{proof}
    See Appendix~\ref{proof: sup and sub, N>2}.
\end{proof}

In the following, we will prove $Q(\mathbf{s}', \mathbf{a}^*) \geq Q(\mathbf{s}', \mathbf{a})$ based on $Q(\mathbf{s}, \mathbf{a}^*) \geq Q(\mathbf{s}, \mathbf{a})$ and similar to the proof of Theorem~\ref{theo:2-s-1-c}, we write the state as $\mathbf{s} = (\tau_{i}, \tau_{j}, \bm{\tau}_{\backslash i,j})$. Using~\eqref{eq:Bellman_Q}, $Q(\mathbf{s}, \mathbf{a}^*) \geq Q(\mathbf{s}, \mathbf{a})$ is equal to
\begin{align}
    p_{i,1} V^{*}(1, \tau_{j}+1, \bm{\tau}_{\backslash i,j}^{+}) \geq p_{j,1} V^{*}(\tau_{i}+1, 1, \bm{\tau}_{\backslash i,j}^{+}) - (p_{j,1}-p_{i,1}) V^{*}(\tau_{i}+1, \tau_{j}+1, \bm{\tau}_{\backslash i,j}^{+}). \label{eq:Q<Q_2}
\end{align}
By applying Lemma~\ref{lemma: probabilistic sup, N>2} to the right-hand side of~\eqref{eq:Q<Q_2}, we have
\begin{equation}
p_{i,1} V^{*}(1, \tau_{j}+1, \bm{\tau}_{\backslash i,j}^{+}) \geq p_{j,1} V^{*}(\tau'_{i}\!+ \!1, 1, \bm{\tau}_{\backslash i,j}^{+})  \!- \! (p_{j,1} \!- \!p_{i,1}) V^{*}(\tau'_{i} \!+ \!1, \tau_{j} \!+ \!1, \bm{\tau}_{\backslash i,j}^{+}), 
\end{equation}
and thus
\begin{align}
    & p_{i,1} V^{*}(1, \tau_{j}\!+\!1, \bm{\tau}_{\backslash i,j}^{+}) \!+\! (1\!-\!p_{i,1}) V^{*}(\tau'_{i}\!+\!1, \tau_{j}\!+\!1, \bm{\tau}_{\backslash i,j}^{+}) \\
    & \quad \geq p_{j,1} V^{*}(\tau'_{i}\!+\!1, 1, \bm{\tau}_{\backslash i,j}^{+}) + (1\!-\!p_{j,1}) V^{*}(\tau'_{i}\!+\!1, \tau_{j}\!+\!1, \bm{\tau}_{\backslash i,j}^{+}),
\end{align}
which is exactly $Q(\mathbf{s}', \mathbf{a}^*) \geq Q(\mathbf{s}', \mathbf{a})$ based on~\eqref{eq:Bellman_Q}.
\end{proof}

\subsubsection{Multi-sensor-multi-channel systems}
This case is more complex than the single-channel one, due to the increased state dimension introduced by the multi-dimensional channel. Thus, the AoI-state threshold property of the optimal policy cannot be derived strictly. In the following, we develop another structural property of the optimal policy.

\begin{proposition}\label{prop: multiple aoi}
If for state $\mathbf{s} = (\bm{\tau}, \mathbf{H})$, the optimal action is $\mathbf{a}^* = (a^*_1, \dots, a^*_i, \dots, a^*_N), a^*_i = m$, and  for state $\mathbf{s}' = (\bm{\tau}'_{(i)}, \mathbf{H}),$ where $ \tau'_{i} \geq \tau_{i}$, the optimal channel assignments except for channel $m$ keep constant, i.e. $a^*_{n} = {a'_{n}}^{*}, \forall a^*_{n} \neq 0,$ and $ i \neq n$, then the optimal action for sensor $j$ at state $\mathbf{s}'$ is to not occupy channel $m$, if sensor $j$ is not scheduled at state $\mathbf{s}$, i.e., ${a_j}^{*} =0$, and its channel condition is worse than that of sensor $i$, i.e., $h_{i,m} \geq h_{j,m}$.
\end{proposition}
Proposition~\ref{prop: multiple aoi} says that sensor $i$'s allocated channel won't be occupied by other sensors with worse channel conditions when sensor $i$'s AoI increases and the optimal channel assignment rules for other channels remain the same.
\begin{proof}
    This proposition equivalent to that if for state $\mathbf{s} = (\bm{\tau}, \mathbf{H})$, the inequality $Q(\mathbf{s}, \mathbf{a}^{*}) \geq Q(\mathbf{s}, \mathbf{a})$ exists, where $a^{*}_{i} = m$ and $ a^{*}_{j} = 0$, then for state $\mathbf{s}' = (\bm{\tau}'_{(i)}, \mathbf{H}),$ where $\tau'_{i} \geq \tau_{i}$, the following inequality holds $Q(\mathbf{s}', \mathbf{a}^{*}) \geq Q(\mathbf{s}', \mathbf{a}')$, where $a'_{j} = m, \forall h_{i,m} \geq h_{j,m}$ and $a^{*}_{\backslash i,j} = a'_{\backslash i,j}$.

    Therefore, we will prove $Q(\mathbf{s}', \mathbf{a}^{*}) \geq Q(\mathbf{s}', \mathbf{a}')$ based on $Q(\mathbf{s}, \mathbf{a}^{*}) \geq Q(\mathbf{s}, \mathbf{a})$ in the following.
    Taking~\eqref{eq: Q i -i} into $Q(\mathbf{s}, \mathbf{a}^{*}) \geq Q(\mathbf{s}, \mathbf{a}')$, we have
    \begin{align}
        & \hspace{-0.5cm}  \sum_{\bm{\tau}_{\backslash i,j}^{+}} \sum_{\mathbf{H}^{+}} \sum_{\tau_{i}^{+}} \sum_{\tau_{j}^{+}}   \operatorname{Pr} (\bm{\tau}_{\backslash i,j}^{+}|\bm{\tau}_{\backslash i,j}, \mathbf{H}_{\backslash i,j}, \mathbf{a}^{*}_{\backslash i,j}) \operatorname{Pr} (\mathbf{H}^{+}) \operatorname{Pr} (\tau_{i}^{+}|\tau_{i},\! \mathbf{h}_{i},\! a^{*}_{i})  \operatorname{Pr} (\tau_{j}^{+}|\tau_{j},\! \mathbf{h}_{j},\! a^{*}_{j}) V^{*}(\mathbf{s}^{+}) \\
        & \hspace{-0.5cm} \!\geq\!  \sum_{\bm{\tau}_{\backslash i,j}^{+}} \sum_{\mathbf{H}^{+}} \sum_{\tau_{i}^{+}} \sum_{\tau_{j}^{+}}   \operatorname{Pr} (\bm{\tau}_{\backslash i,j}^{+}|\bm{\tau}_{\backslash i,j},\! \mathbf{H}_{\backslash i,j},\! \mathbf{a}'_{\backslash i,j}) \! \operatorname{Pr} (\mathbf{H}^{+}) \!  \operatorname{Pr} (\tau_{i}^{+}|\tau_{i},\! \mathbf{h}_{i},\! a'_{i}) \! \operatorname{Pr} (\tau_{j}^{+}|\tau_{j},\! \mathbf{h}_{j},\! a'_{j}) V^{*}(\mathbf{s}^{+}). \label{eq: pro2 Q1<Q2}
    \end{align}
Based on $a^{*}_{\backslash i,j} = a'_{\backslash i,j}$, we have
\begin{equation}\label{eq: P(-i,-j)}
    \operatorname{Pr} (\bm{\tau}_{\backslash i,j}^{+}|\bm{\tau}_{\backslash i,j}, \mathbf{H}_{\backslash i,j}, \mathbf{a}^{*}_{\backslash i,j}) 
    = \operatorname{Pr} (\bm{\tau}_{\backslash i,j}^{+}|\bm{\tau}_{\backslash i,j}, \mathbf{H}_{\backslash i,j}, \mathbf{a}'_{\backslash i,j}).
\end{equation}
Thus, from~\eqref{eq: pro2 Q1<Q2}, the following inequality can be derived
\begin{align}
    \!\!\sum_{\tau_{i}^{+}} \!\sum_{\tau_{j}^{+}}  \operatorname{Pr}\! \left(\tau_{i}^{+}|\tau_{i},\! \mathbf{h}_{i},\! a^{*}_{i}\right) \! 
     \operatorname{Pr} \!\left(\tau_{j}^{+}|\tau_{j},\! \mathbf{h}_{j},\! a^{*}_{j}\right) \!\! V^{*}\!\!\left(\mathbf{s}^{+}\right) 
     \!\geq\! \sum_{\tau_{i}^{+}}\!\sum_{\tau_{j}^{+}} \operatorname{Pr} \! \left(\tau_{i}^{+}|\tau_{i},\! \mathbf{h}_{i},\! a'_{i}\right) \!
     \operatorname{Pr} \! \left(\tau_{j}^{+}|\tau_{j},\! \mathbf{h}_{j},\! a'_{j}\right) \!\! V^{*}\!\!\left(\mathbf{s}^{+}\right). \!\!
\end{align}
For notation simplicity, the state can be rewritten as $\mathbf{s} = (\tau_{i}, \tau_{j})$, because the states except for $\tau_{i}$ and $\tau_{j}$ are constant.
Based on $a^{*}_{i} = m, a^{*}_{j} = 0,a'_{i} = 0,$ and $a'_{j} = m$, the left-hand side and the right hand side of~\eqref{eq: pro2 Q1<Q2} are equal to 
\begin{align}
\begin{cases}
    p_{i,m} \! V^{*}(1, \!\tau_{j}\!+\!\!1) \!+\! (1\!\!-\!p_{i,m}) V^{*}(\tau_{i}\!+\!\!1, \!\tau_{j}\!+\!\!1)\\
    p_{j,m} \!V^{*}(\tau_{i}\!+\!\!1,\! 1) \!+\! (1\!\!-\!p_{j,m}) V^{*}(\tau_{i}\!+\!\!1, \!\tau_{j}\!+\!\!1). \label{eq: P(i,j),a}
\end{cases}
\end{align}
Then, the following inequality can be derived
\begin{align}
    p_{i,m} V^{*}(1, \tau_{j}+1) 
    & \!\geq p_{j,m} V^{*}(\tau_{i}+1, 1) + (p_{i,m} - p_{j,m}) V^{*}(\tau_{i}+1, \tau_{j}+1) \\
    \label{eq: prop1 P(i) > P(j)}
    & \!\geq p_{j,m} V^{*}(\tau'_{i}+1, 1) + (p_{i,m} - p_{j,m}) V^{*}(\tau'_{i}+1, \tau_{j}+1), 
\end{align}
where the first inequality is derived by~\eqref{eq: pro2 Q1<Q2},~\eqref{eq: P(-i,-j)}, and~\eqref{eq: P(i,j),a}, and the second inequality is by using Lemma~\ref{lemma:monotone} and $p_{i,m} \geq p_{j,m}$.
The inequality~\eqref{eq: prop1 P(i) > P(j)} can be rewritten as
\begin{align}\label{eq:Vi(tau'_i)>Vj(tau'_j)}
     p_{i,m} V^{*}(1, \!\tau_{j}\!+\!\!1) \!+\! (1 \!-\! p_{i,m}) V^{*}(\tau'_{i} \!+\!\!1, \!\tau_{j}\!+\!\!1) \!\geq\! p_{j,m} V^{*}(\tau'_{i}\!+\!\!1, \!1) \!+\! (1 \!-\! p_{j,m}) V^{*}(\tau'_{i}\!+\!\!1, \!\tau_{j}\!+\!\!1). \hspace{-0.2cm}
\end{align}
According to the analysis from~\eqref{eq: pro2 Q1<Q2} to~\eqref{eq: P(i,j),a}, it directly follows $Q(\mathbf{s}', \mathbf{a}^{*}) \geq Q(\mathbf{s}', \mathbf{a}')$ from~\eqref{eq:Vi(tau'_i)>Vj(tau'_j)}.
\end{proof}

\begin{remark}[Generality of the derived results]\label{remark: different reward function}
We note that all the theoretical results proved in this section merely rely on two features of the formulated optimal scheduling problem: 1) the reward function of each user is a monotonically increasing function in terms of its AoI state, and 2) the total reward function of the problem is the sum of individual rewards.
In other words, for other scheduling problems of different systems, if the two features are satisfied, the structural properties still hold.
For example, if a scheduling problem is for minimizing the overall expected sum AoI (i.e., the instantaneous reward is $r(\mathbf{s}_{t}) = \sum_{n=1}^{N} -\tau_{n,t}$), not the sum estimation MSE, then the optimal policy does have the same structural properties.
\end{remark}

\begin{remark}[Counterexample]\label{remark:counterexample}
We can show that if the considered optimal scheduling problem of a multi-sensor remote estimation system has a reward function as the negative product of the individual estimation MSE, i.e., $r(\mathbf{s}) = -\prod_{i=1}^{N} \operatorname{Tr} (\mathbf{P}_{i,t})$, which does not satisfy the second feature in Remark~\ref{remark: different reward function}, then the threshold structure may not exist. Fig~\ref{fig:counterexample} illustrates the optimal scheduling policy of a two-sensor-single-channel remote estimation system with the constructed reward function. We see that the optimal policy is not a threshold policy as in Theorem~\ref{theo:2-s-1-c}.
\end{remark}
\begin{figure}[t]
    \centering
        \centering
        \includegraphics[width=0.45\textwidth]{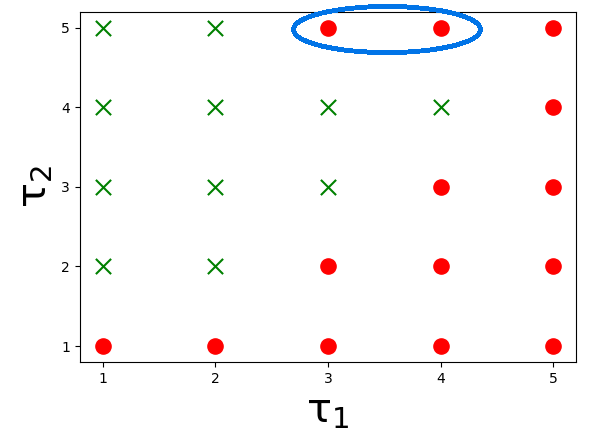}
        \vspace{-0.2cm}
        \caption{The optimal policy of a two-sensor-single-channel scheduling problem with the multiplicative reward function, where $\bullet$ and $\times$ represent the schedule of sensor 1 and 2, respectively.}
        \label{fig:counterexample}
    \vspace{-0.7cm}
\end{figure}

\section{Structure-Enhanced DRL} \label{sec: DRL}
In the literature, DQN and DDPG are the most commonly adopted off-policy DRL algorithms for solving optimal scheduling problems (see \cite{leong2020DRL,pang2022drl} and references therein), and they provide significant performance improvements over heuristic policies.
Next, we develop threshold structure-enhanced (SE) DQN and DDPG for solving Problem~\ref{pro1} based on the structural properties in Definitions~\ref{def: threshold-channel} and~\ref{def: threshold-AoI}. The performance improvement of SE algorithms will be presented in Section~\ref{sec: simulation}.

\subsection{Structure-Enhanced DQN}
Based on~\eqref{eq:V} and~\eqref{eq:Bellman_Q}, we have the Bellman equation in terms of the Q-values of the optimal policy:
\begin{equation}\label{eq:optimal Q}
    Q(\mathbf{s}_t, \mathbf{a}_t) = r(\mathbf{s}_t) + \mathbb{E}_{\mathbf{s}_{t+1}} \left[ \gamma \max_{\mathbf{a}_{t+1}}Q \left(\mathbf{s}_{t+1}, \mathbf{a}_{t+1}\right) \right].
\end{equation}
A well-trained DQN uses a neural network (NN) with the parameter set $\bm{\theta}$ to approximate $Q(\mathbf{s}_{t}, \mathbf{a}_{t})$ by $Q(\mathbf{s}_{t}, \mathbf{a}_{t};\bm{\theta})$ and use it to find the optimal action, i.e., $\mathbf{a}^*_t = \operatorname{argmax}_{\mathbf{a}_{t}\in\mathcal{A}}Q(\mathbf{s}_{t}, \mathbf{a}_{t};\bm{\theta})$. 
Considering the action space defined in Section~\ref{sec:MDP}, the DQN has $N!/(N-M)!$ Q-value outputs of different state-action pairs.
To train the DQN, one needs to sample data (consisting of states, actions, rewards, and next states), define a loss function of $\bm \theta$ based on the collected data, and minimize the loss function to find the optimized $\bm \theta$. 
However, the conventional DQN training method has never utilized structures of optimal policies before.

To utilize the knowledge of the threshold policy structure for enhancing the DQN training performance, we propose \textbf{1) an SE action selection method} based on Definitions~\ref{def: threshold-channel} and~\ref{def: threshold-AoI} to select reasonable actions and hence  enhance the data sampling efficiency; 
and \textbf{2) an SE loss function definition} to add the penalty to sampled actions that do not follow the threshold structure.

Our SE-DQN training algorithm has three stages: 1) the DQN with \textbf{loose SE action selection} stage, which only utilizes part of the structural property,  2) the DQN with \textbf{tight SE action selection} stage utilizes the full structural property, and 3) the conventional DQN stage. 
The first two stages use the SE action selection schemes and the SE loss function to train the DQN fast, resulting in a reasonable threshold policy, and the last stage is for further policy exploration.
In what follows, we present the loose and tight SE action selection schemes and the SE loss function.

\subsubsection{Loose SE action selection}
We randomly select an action $\mathbf{a}^\epsilon$ from the entire action space with a probability of $\epsilon$ for action exploration; with a probability of $(1-\epsilon)$, we generate the SE action $\hat{\mathbf{a}}$ as below. For simplicity, we drop the time index when describing action selections.

The threshold structure suggests that the actions of $\mathbf{s}$ and the state with a smaller AoI or channel state are correlated. 
Thus, one can infer the action based on the action at the state with a smaller channel, or AoI state, based on the channel-state and the AoI-state threshold properties in Definitions~\ref{def: threshold-channel} and~\ref{def: threshold-AoI}, respectively.
We only consider the AoI-state threshold property for loose SE action selection, as it is difficult to find actions that satisfy both structural properties at the beginning of training.  We will utilize both of them in the tight SE action selection stage.

Given the state $\mathbf{s} = (\bm{\tau}, \mathbf{H})$, we define the state with a smaller sensor $n$'s AoI as $\dot{\mathbf{s}}^{(n)} = (\dot{\bm{\tau}}^{(n)}, \mathbf{H})$, where 
\begin{equation}\label{eq: last AoI}
    \dot{\bm{\tau}}^{(n)} = (\tau_1, \dots, \tau_{n}-1, \dots, \tau_N).
\end{equation}
For each $n$, we calculate the corresponding action based on the Q-values as
\begin{align}\label{eq:AoI action}
    \dot{\mathbf{a}}^{(n)} \triangleq  (\dot{a}^{(n)}_{1}, \dots, \dot{a}^{(n)}_{n}, \dots, \dot{a}^{(n)}_{N})= \mathop{\arg\max}_{\dot{\mathbf{a}}} Q(\dot{\mathbf{s}}^{(n)}, \dot{\mathbf{a}}; \bm{\theta}).\
\end{align}
Recall that $\dot{a}^{(n)}_{n}\in\{0,1,\dots,M\}$ is the channel index assigned to sensor $n$ at the state $\dot{\mathbf{s}}^{(n)}$.

If $\dot{a}^{(n)}_{n}>0$, then the AoI-state threshold property implies that channel $\dot{a}^{(n)}_{n}$ or a better channel is assigned to sensor $n$ at the state $\mathbf{s}$. We define the set of channels with better quality as
\begin{equation}\label{eq:set M}
\mathcal{M}^{(n)} \triangleq \{m'|h_{n,m'} > h_{n,\dot{a}^{(n)}_{n}}, m'=1,\dots,M \}.
\end{equation}
Then, the SE action for sensor $n$, say $\hat{a}_n$, is randomly chosen from the set $\mathcal{M}^{(n)}$ with probability~$\xi$, and is equal to $\dot{a}^{(n)}_{n}$ with probability $1-\xi$.

If $\dot{a}^{(n)}_{n}=0$, then the AoI-state threshold property cannot help with determining the action. Then, we define the action selected by the greedy policy (i.e., the conventional DQN method) at the state $\mathbf{s}$ as
\begin{equation}\label{eq: greedy action}
\tilde{\mathbf{a}} \triangleq (\tilde{a}_1,\tilde{a}_2,\dots,\tilde{a}_N) = \mathop{\arg\max}_{\mathbf{a}} Q(\mathbf{s}, \mathbf{a}; \bm{\theta}).
\end{equation}
Thus, we set the SE action of sensor $n$ identical to the one generated by the conventional DQN method, i.e., $\hat{a}_n=\tilde{a}_n$.

Now we can define the SE action for $N$ sensors as $\hat{\mathbf{a}} = (\hat{a}_{1}, \hat{a}_{2}, \dots, \hat{a}_{N})$. If such an action meets the constraint
\begin{align}\label{eq: loose action constraint}
    \sum_{m=1}^{M} \boldsymbol{\mathbbm{1}}\left( \hat{a}_{n} = m \right) \leq 1, \quad
    \sum_{n=1}^{N} \boldsymbol{\mathbbm{1}}\left( \hat{a}_{n} = m \right) \leq 1,
\end{align}
which is less restrictive than \eqref{eq: action constraint},
we select the action $\mathbf{a}$ as $\hat{\mathbf{a}}$ and assign the unused channels randomly to unscheduled sensors; otherwise, $\mathbf{a}$ is identical to that of the conventional method as~$\tilde{\mathbf{a}}$. 

\subsubsection{Tight SE action selection} By using the loose SE action selection, we first infer the scheduling action of sensor $n$ at the state $\mathbf{s}$ based on the action of the state with a smaller AoI, $\dot{\mathbf{s}}^{(n)}$. Then, we check whether the loose SE action satisfies the channel-state threshold property as below.

For notation simplicity, we use $m>0$ to denote the SE channel selection for sensor $n$. 
Given the state $\mathbf{s}$, we define the state with a smaller channel $m$'s state for sensor $n$ as $\ddot{\mathbf{s}}^{(n,m)} = (\bm{\tau}, \ddot{\mathbf{H}}^{(n,m)})$, where $\ddot{\mathbf{H}}^{(n,m)}$ and $\mathbf{H}$ are identical except the element $\ddot{h}^{(n,m)}_{n,m} = h_{n,m} -1$. Then, we calculate the corresponding action
\begin{equation} \label{eq:channel action}
    \ddot{\mathbf{a}}^{(n,m)} \!\!\triangleq\! (\ddot{a}^{(n,m)}_{1}\!\!\!, \!\dots,\! \ddot{a}^{(n,m)}_{n}\!\!\!, \dots, \ddot{a}^{(n,m)}_{N}) 
    \!=\! \mathop{\arg\max}_{\ddot{\mathbf{a}}} Q(\ddot{\mathbf{s}}^{(n,m)}\!\!\!, \ddot{\mathbf{a}};\! \bm{\theta}). 
\end{equation}
From the channel-state and the AoI-state threshold properties, the scheduling action $\ddot{a}^{(n,m)}_{n}$ should be identical to $m$. 
Thus, if $\ddot{a}^{(n,m)}_{n}=m$, then the SE action for sensor $n$ is $\hat{a}_n=m$; otherwise, $\hat{a}_n$ is identical to the conventional DQN action.
The SE action satisfying both the threshold properties is $\hat{\mathbf{a}} = (\hat{a}_{1}, \hat{a}_{2}, \dots, \hat{a}_{N})$. If such an action meets the constraint~\eqref{eq: action constraint}, 
then it is executed as $\mathbf{a} = \hat{\mathbf{a}}$; otherwise, we select the greedy action $\mathbf{a} = \tilde{\mathbf{a}}$. 

\subsubsection{SE loss function}
During the training, each transition $(\mathbf{s}_t, \mathbf{a}_t, \hat{\mathbf{a}}_t, \tilde{\mathbf{a}}_{t}, r_t, \mathbf{s}_{t+1})$ is stored in a replay memory, where $r_t\triangleq r(\mathbf{s}_t)$ denotes the immediate reward.
Different from the conventional DQN, we include both the SE action $\hat{\mathbf{a}}$ and the greedy action $\tilde{\mathbf{a}}$, in addition to the executed action $\mathbf{a}$.

\begin{algorithm}[t]
    \small
    \caption{\small{SE-DQN for sensor scheduling in the  remote estimation system}}
    \label{alg:SE-DQN}
    \begin{algorithmic}[1]
        \State Initialize replay memory $\mathcal{D}$ to capacity $N$
        \State Initialize policy network with random weights $\bm{\theta}$
        \State Initialize target network with weights $\hat{\bm{\theta}} = \bm{\theta}$
        \For {episode $= 1, 2, \dots, E_{1}$}
            \State Initialize state $\mathbf{s}_0$
            \For {$t = 1, 2, \dots, T$}
                \State Generate $\hat{\mathbf{a}}_t$ and select action $\mathbf{a}_t$ by the loose SE action selection method of SE-DQN
                \State Execute action $\mathbf{a}_t$ and observe $r_t$ and $\mathbf{s}_{t+1}$
                \State Compute action $\tilde{\mathbf{a}}_{t} = \mathop{\arg\max}_{\mathbf{a}} Q(\mathbf{s}_{t}, \mathbf{a}; \bm{\theta})$
                \State Store transition $(\mathbf{s}_t, \mathbf{a}_t, \hat{\mathbf{a}}_t, \tilde{\mathbf{a}}_{t}, r_t, \mathbf{s}_{t+1})$ in $\mathcal{D}$
                \State Sample a random batch of transitions $(\mathbf{s}_{i}, \mathbf{a}_{i}, \hat{\mathbf{a}}_{i}, \tilde{\mathbf{a}}_{i}, r_{i}, \mathbf{s}_{i+1})$
                \State Calculate $\mathsf{TD}_i$ and $\mathsf{AD}_i$ based on \eqref{eq:TD} and \eqref{eq:AD}
                \State Update $\bm \theta$ according to equation \eqref{grat:critic}
                \State Every $t'$ steps set $\hat{\bm{\theta}} = \bm{\theta}_{t}$
            \EndFor
        \EndFor
        \For {episode $= E_{1}, \dots, E_{2}$}
            \State Repeat algorithm from line 5 to line 15 by adopting the tight SE action selection method
        \EndFor
        \For {episode $= E_{2}, \dots, E$}
            \State Execute the DQN algorithm as in \cite{leong2020DRL}
        \EndFor
    \end{algorithmic}
\end{algorithm}

Let $\mathcal{T}_{i} \triangleq (\mathbf{s}_i, \mathbf{a}_i, \hat{\mathbf{a}}_i, \tilde{\mathbf{a}}_{t}, r_i, \mathbf{s}_{i+1})$ denote the $i$th transition of a sampled batch from the replay memory. 
Same as the conventional DQN, we define the temporal-difference (TD) error of $\mathcal{T}_i$ as
\begin{equation}\label{eq:TD}
    \mathsf{TD}_{i} \triangleq y_{i}-Q(\mathbf{s}_{i},\mathbf{a}_{i};\bm{\theta}),
\end{equation} 
where $y_{i} = r_{i} + \gamma \max_{{\mathbf{a}}_{i+1}} Q(\mathbf{s}_{i+1},{\mathbf{a}}_{i+1};\bm{\theta})$ is the estimation of Q-value at next step.
This is to measure the gap between the left and right sides of the Bellman equation~\eqref{eq:optimal Q}. A larger gap indicates that the approximated Q-values are far from the optimal.
Different from DQN, we introduce the action-difference (AD) error as below to measure the difference of Q-value between actions selected by the SE strategy and the greedy strategy:
\begin{equation}\label{eq:AD}
    \mathsf{AD}_{i} \triangleq Q(\mathbf{s}_{i},\hat{\mathbf{a}}_{i};\bm{\theta}) - Q(\mathbf{s}_{i}, \tilde{\mathbf{a}}_{i};\bm{\theta}).
\end{equation}
Since the optimal policy has the threshold structure, the inferred action $\hat{\mathbf{a}}$ should be identical to the optimal action $\tilde{\mathbf{a}}$. Thus, a well-trained $\bm \theta$ should lead to a small difference in \eqref{eq:AD}. 

Based on \eqref{eq:TD} and \eqref{eq:AD}, we define the loss function of $\mathcal{T}_i$ as \vspace{-0.5cm}
\begin{equation}\label{loss:critic_i}
    L_{i}(\bm{\theta}) = \left\{
    \begin{array}{l}
    \begin{aligned}
        & \alpha_1 \mathsf{TD}^{2}_{i} + (1-\alpha_1) \mathsf{AD}^{2}_{i},  \quad \text{if $\mathbf{a}_{i} = \hat{\mathbf{a}}_{i}$} \\
        & \mathsf{TD}^{2}_{i}, \qquad \qquad \qquad \qquad \ \;  \text{otherwise},
    \end{aligned}
    \end{array}
    \right. 
\end{equation}
where $\alpha_1$ is a hyperparameter to balance the importance of $\mathsf{TD}_i$ and $\mathsf{AD}_i$.
In other words, if the SE action is executed, both the TD and AD errors are taken into account; otherwise, the conventional TD-error-based loss function is adopted.

Given the batch size $B$, the overall loss function is
\begin{equation}\label{loss:critic}
    L(\bm{\theta}) = \frac{1}{B} \sum_{i=1}^{B} L_{i}(\bm{\theta}).
\end{equation}
To optimize $\bm \theta$, we adopt the well-known gradient descent method and calculate the gradient as below
\begin{equation}\label{grat:critic}
    \nabla_{\bm{\theta}} L(\bm{\theta}) = \frac{1}{B} \sum_{i=1}^{B} \nabla_{\bm{\theta}} L_{i}(\bm{\theta}),
\end{equation}
where $\nabla_{\bm{\theta}} L_{i}(\bm{\theta})$ is given as
\begin{equation}\label{grat:critic_i}
    \nabla_{\bm{\theta}} L_{i}(\bm{\theta}) = \left\{
    \begin{array}{l}
    \begin{aligned}
        & \!\!\!-2 ((1-\alpha_1) \mathsf{AD}_{i} \nabla_{\bm{\theta}}\left(Q\left(\mathbf{s}_{i},\hat{\mathbf{a}}_{i};\bm{\theta}) - Q(\mathbf{s}_{i}, \tilde{\mathbf{a}}_{i};\bm{\theta}\right)\right) \\
        & \ \ +\alpha_1\mathsf{TD}_{i} \nabla_{\bm{\theta}} Q(\mathbf{s}_{i}, \mathbf{a}_{i};\bm{\theta})), \quad \text{if $\mathbf{a}_{i} = \hat{\mathbf{a}}_{i}$} \\
        & \!\!\! -2(\mathsf{TD}_{i} \nabla_{\bm{\theta}} Q(\mathbf{s}_{i}, \mathbf{a}_{i};\bm{\theta})), \qquad \; \, \text{otherwise}.
    \end{aligned}
    \end{array}
    \right. 
\end{equation}

The details of the SE-DQN algorithm are given in Algorithm~\ref{alg:SE-DQN}. 


\subsection{Structure-Enhanced DDPG}
Different from the DQN, which has one NN to estimate the Q-value, a DDPG agent has two NNs~\cite{lillicrap2015DDPG}, a critic NN with parameter $\bm \theta$ and an actor NN with the parameter $\bm \mu$.
In particular,  the actor NN   approximates the optimal policy $\pi^*(\mathbf{s})$ by $\mu(\mathbf{s};\bm{\mu})$, while the critic NN approximates the Q-value of the optimal policy $Q(\mathbf{s}, \mathbf{a})$ by $Q(\mathbf{s}, \mathbf{a};\bm{\theta})$.
In general, the critic NN judges whether the generated action of the actor NN is good or not, and the latter can be improved based on the former. The critic NN is updated by minimizing the TD error similar to the DQN.
The actor-critic framework enables DDPG to solve MDPs with continuous and large action space, which cannot be handled by the DQN.

To solve our scheduling problem with a discrete action space, we adopt an action mapping scheme similar to the one adopted in~\cite{pang2022drl}. 
We set the direct output of the actor NN with $N$ continuous values, ranging from $-1$ to $1$, corresponding to sensors $1$ to $N$, respectively. Recall that the DQN has $N!/(N-M)!$ outputs.
The $N$ values are sorted in descending order. The sensors with the highest $M$ ranking are assigned to channels $1$ to $M$, respectively. The corresponding ranking values are then linearly normalized to $[-1,1]$ as the output of the final outputs of the actor NN, named as the virtual action $\mathbf{v}$. 
It directly follows that the virtual action $\mathbf{v}$ and the real scheduling action $\mathbf{a}$ can be mapped from one to the other directly. 
Therefore, we use the virtual action $\mathbf{v}$, instead of the real action $\mathbf{a}$, when presenting the SE-DDPG algorithm.

Similar to the SE-DQN, the SE-DDPG has the loose SE-DDPG stage, the tight SE-DDPG stage, and the conventional DDPG stage. The $i$th sampled transition is denoted as
\begin{equation}
\mathcal{T}_{i} \triangleq (\mathbf{s}_i, \mathbf{v}_i, \hat{\mathbf{v}}_i, \tilde{\mathbf{v}}_{t}, r_i, \mathbf{s}_{i+1}).
\end{equation} 
We present the SE action selection method and the SE loss function in the sequel.

\subsubsection{SE action selection}\label{subsec:SE-DDPG action} The general action selection approach for the SE-DDPG is identical to that of the SE-DQN, by simply converting $\mathbf{v}_i$, $\mathbf{v}^\epsilon_i$, $\tilde{\mathbf{v}}_i$, and $\hat{\mathbf{v}}_i$ to $\mathbf{a}_i$, $\mathbf{a}^\epsilon_i$, $\tilde{\mathbf{a}}_i$, and $\hat{\mathbf{a}}_i$, respectively. 
Different from DQN with $\epsilon$-greedy actions, the action generated by the DDPG based on the current state is 
\begin{equation}\label{eq:tilde_v}
\tilde{\mathbf{v}}_i = \mu(\mathbf{s}_i;\bm{\mu}),
\end{equation}
and the random action $\mathbf{v}^\epsilon_i$ was generated by adding noise to the original continuous output of the actor NN.

\subsubsection{SE loss function} Different from the SE-DQN, the SE-DDPG needs different loss functions for updating the critic NN and the actor NN.
For the critic NN, we use the same loss function as in \eqref{loss:critic}, and thus the gradient for the critic NN update is identical to \eqref{grat:critic}.
Note that for DDPG, the next virtual action $\tilde{\mathbf{v}}_{i+1}$ is the direct output of the actor NN given the next state $\mathbf{s}_{i+1}$, i.e., $\tilde{\mathbf{v}}_{i+1} = \mu(\mathbf{s}_{i+1};\bm \mu)$. Thus, when calculating the TD error \eqref{eq:TD}, we have $y_{i} = r_{i} + \gamma Q(\mathbf{s}_{i+1}, {\mu}(\mathbf{s}_{i+1};\bm\mu);\bm{\theta})$.

For the actor NN, we introduce the difference between actions selected by the SE strategy $\hat{\mathbf{v}}_i$ and the actor NN  $\tilde{\mathbf{v}}_i$, when $\hat{\mathbf{v}}_i$ is executed, i.e., $\mathbf{v}_i=\hat{\mathbf{v}}_i$. If the SE action is not selected, then the loss function is the Q-value given the state-action pair, which is identical to the conventional DDPG. Given the hyperparameter $\alpha_2$, the loss function for the transition~$\mathcal{T}_i$ is defined as
\begin{equation}\label{loss:actor_i}
    L_{i}(\bm{\mu}) = \left\{
    \begin{array}{l}
    \begin{aligned}
        & \alpha_2 Q(\mathbf{s}_{i}, \tilde{\mathbf{v}}_{i}; \bm{\theta}) + (1-\alpha_2) \left( \mathbf{v}_{i} - \tilde{\mathbf{v}}_{i} \right)^{2},  \ \ \text{if $\mathbf{v}_{i} = \hat{\mathbf{v}}_{i}$} \\
        & Q(\mathbf{s}_{i}, \tilde{\mathbf{v}}_{i}; \bm{\theta}), \qquad \qquad \qquad \qquad \qquad \ \ \  \text{otherwise},
    \end{aligned}
    \end{array}
    \right. 
\end{equation}
\begin{algorithm}[t]
\small
    \caption{\small{SE-DDPG for sensor scheduling in the  remote estimation system}}
    \label{alg:SE-DDPG}
    \begin{algorithmic}[1]
        \State Initialize replay memory $\mathcal{D}$ to capacity $N$
        \State Initialize actor network and critic network with random weights $\bm{\mu}$ and $\bm{\theta}$
        \State Initialize target network and with weight $\hat{\bm{\mu}} = \bm{\mu}$, $\hat{\bm{\theta}} = \bm{\theta}$
        \For {episode $= 1, 2, \dots, E_{1}$}
            \State Initialize a random noise $\mathcal{N}$ for action exploration
            \State Initialize state $\mathbf{s}_0$
            \For {$t=1, 2, \dots, T$}
                \State Generate $\hat{\mathbf{v}}_{t}$ and select action $\mathbf{v}_{t}$ by the loose SE action selection method of SE-DDPG
                \State Mapping virtual action $\mathbf{v}_{t}$ to real action $\mathbf{a}_{t}$
                \State Execute action $\mathbf{a}_{t}$ and observe $r_t$ and $\mathbf{s}_{t+1}$
                \State Compute action $\tilde{\mathbf{v}}_{t} = \mu(\mathbf{s}_t;\bm{\mu})$
                \State Store transition $(\mathbf{s}_t, \mathbf{v}_t, \hat{\mathbf{v}}_t, \tilde{\mathbf{v}}_{t}, r_t, \mathbf{s}_{t+1})$ in $\mathcal{D}$
                \State Sample a random batch of transitions $(\mathbf{s}_{i}, \mathbf{v}_{i}, \hat{\mathbf{v}}_{i}, \tilde{\mathbf{v}}_{i}, r_{i}, \mathbf{s}_{i+1})$
                \State Calculate $\mathsf{TD}_i$ and $\mathsf{AD}_i$ based on \eqref{eq:TD} and \eqref{eq:AD} but with the virtual actions $\mathbf{v}_{i}, \hat{\mathbf{v}}_{i}, \Tilde{\mathbf{v}}_{i}$
                \State Update $\bm{\theta}$ according to equation \eqref{grat:critic}
                \State Update $\bm{\mu}$ according to equation \eqref{grat:actor}
                \State Update the target network:
                \begin{align}
                    \hat{\bm{\mu}} & \gets \delta \bm{\mu} + (1 - \delta)\hat{\bm{\mu}} \\
                    \hat{\bm{\theta}} & \gets \delta \bm{\theta} + (1 - \delta)\hat{\bm{\theta}}
                \end{align}
            \EndFor
        \EndFor
        \For {episode $E_{1} = 1, 2, \dots, E_{2}$}
            \State Repeat algorithm from line 5 to line 18 by adopting the tight SE action selection method
        \EndFor
        \For {episode $E_{2} = 1, 2, \dots, E$}
            \State Execute the conventional DDPG algorithm as in \cite{lillicrap2015DDPG}
        \EndFor
    \end{algorithmic}
\end{algorithm}
and hence the overall loss function given the sampled batch is 
\begin{equation}\label{loss:actor}
    L(\bm{\mu}) = \frac{1}{B} \sum_{i=1}^{B} L_{i}(\bm{\mu}).
\end{equation}
By replacing \eqref{eq:tilde_v} and \eqref{loss:actor_i} in \eqref{loss:actor} and applying the chain rule, we can derive the gradient of the overall loss function in terms of $\bm \mu$ as
\vspace{-0.2cm}
\begin{equation}\label{grat:actor}
    \nabla_{\bm{\mu}} L(\bm{\mu}) = \frac{1}{B} \sum_{i=1}^{B} \nabla_{\bm{\mu}} L_{i}(\bm{\mu}),
\end{equation}
\vspace{-0.3cm}
where $\nabla_{\bm{\mu}} L_{i}(\bm{\mu})$ is given by:
\begin{equation}\label{grat:actor_i}
    \nabla_{\bm{\mu}} L_{i}(\bm{\mu}) = \left\{
    \begin{array}{l}
    \begin{aligned}
        & \alpha_2 \nabla_{\tilde{\mathbf{v}}_{i}} Q(\mathbf{s}_{i}, \tilde{\mathbf{v}}_{i}; \bm{\theta}) \nabla_{\bm{\mu}} \mu(\mathbf{s}_{i};\bm{\mu}) - 2 (1-\alpha_2) \\
        & \ \ \left( \mathbf{v}_{i} - \mu\left(\mathbf{s}_{i};\bm{\mu}\right) \right) \nabla_{\bm{\mu}} \mu(\mathbf{s}_{i};\bm{\mu}) , \quad \text{if $\mathbf{v}_{i} = \hat{\mathbf{v}}_{i}$} \\
        & \nabla_{\tilde{\mathbf{v}}_{i}} Q(\mathbf{s}_{i}, \tilde{\mathbf{v}}_{i}; \bm{\theta}) \nabla_{\bm{\mu}} \mu(\mathbf{s}_{i};\bm{\mu}), \qquad \text{otherwise}.
    \end{aligned}
    \end{array}
    \right. 
\end{equation}
\vspace{-0.5cm}
The details of the proposed  SE-DDPG algorithm are given in Algorithm~\ref{alg:SE-DDPG}.

\section{Numerical Experiments} \label{sec: simulation}

In this section, we evaluate and compare the performance of the proposed SE-DQN and SE-DDPG with the benchmark DQN and DDPG.
\subsection{Experiment Setups}
\begin{table}[t]
    \begin{minipage}[b]{0.51\textwidth}
    \footnotesize
    \setlength\tabcolsep{1pt}
    \centering
    \caption{Summary of Hyperparameters}
    \label{tab:Setup}
    \vspace{-0.3cm}
    \begin{tabular}{c|c}
         \hline \hline
         \textbf{Hyperparameters of SE-DQN and SE-DDPG} & Value \\
         \hline
         Initial values of $\epsilon$ and $\xi$ & 1 \\
         \rowcolor[HTML]{EFEFEF} 
         Decay rates of $\epsilon$ and $\xi$                     & 0.999 \\
         Minimum $\epsilon$ and $\xi$                        & 0.01 \\
         \rowcolor[HTML]{EFEFEF} 
         Mini-batch size, $B$                       & 128 \\
         Experience replay memory size, $K$         & 20000 \\
         \rowcolor[HTML]{EFEFEF} 
         Discount factor, $\gamma$                  & 0.95 \\
         \hline
         \textbf{Hyperparameters of SE-DQN} \\
         \hline
         Learning rate                              & 0.0001 \\
         \rowcolor[HTML]{EFEFEF}
         Decay rate of learning rate                & 0.001 \\
         Target network update frequency            & 100 \\
         \rowcolor[HTML]{EFEFEF}
         Weight of SE-DQN loss function, $\alpha_1$ & 0.5 \\
         Input dimension of network   & $N+N \times M$ \\
         \rowcolor[HTML]{EFEFEF}
         Output dimension of network     & $N!/(N-M)!$ \\
         \hline
         \textbf{Hyperparameters of SE-DDPG} \\
         \hline
         Learning rate of actor network             & 0.0001 \\
         \rowcolor[HTML]{EFEFEF}
         Learning rate of critic network             & 0.001 \\
         Decay rate of learning rate                & 0.001 \\
         \rowcolor[HTML]{EFEFEF}
         Soft parameter for target update, $\delta$ & 0.005 \\
         Weight of critic network loss function, $\alpha_1$ & 0.5 \\
         \rowcolor[HTML]{EFEFEF}
         Weight of actor network loss function, $\alpha_2$ & 0.9 \\
         Input dimension of actor network   & $N+N \times M$ \\
         \rowcolor[HTML]{EFEFEF}
         Output dimension of actor network   & $N$ \\
         Input dimension of critic network   & $2N+N \times M$ \\
         \rowcolor[HTML]{EFEFEF}
         Output dimension of critic network   & $1$ \\
         \hline \hline
    \end{tabular}
    \end{minipage}
    \hfill
    \begin{minipage}[b]{0.47\textwidth}
        \centering
        \includegraphics[width=1\textwidth]{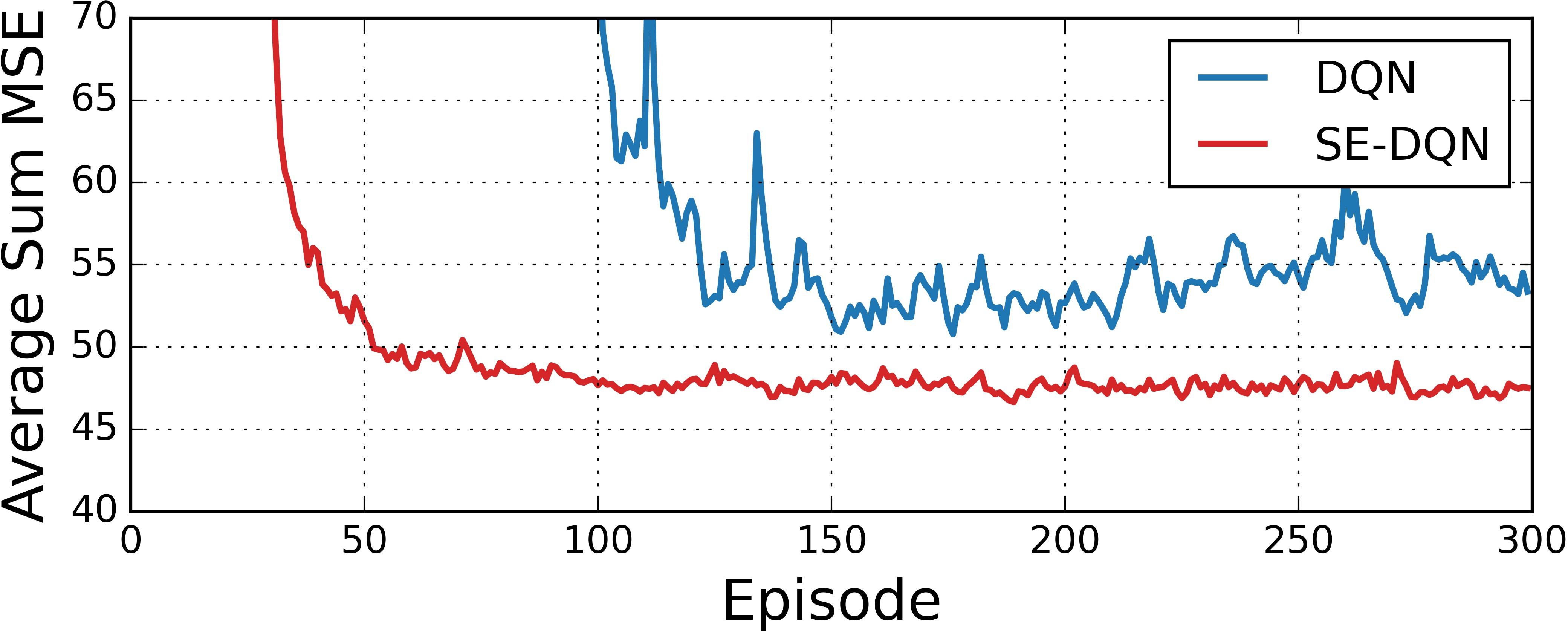}
        \vspace{-0.6cm}
        \captionof{figure}{Average sum MSE of all processes during training with $N=6, M=3$.}
        \label{fig:training curve DQN}
        \vspace{0.3cm}
        \centering
        \includegraphics[width=1\textwidth]{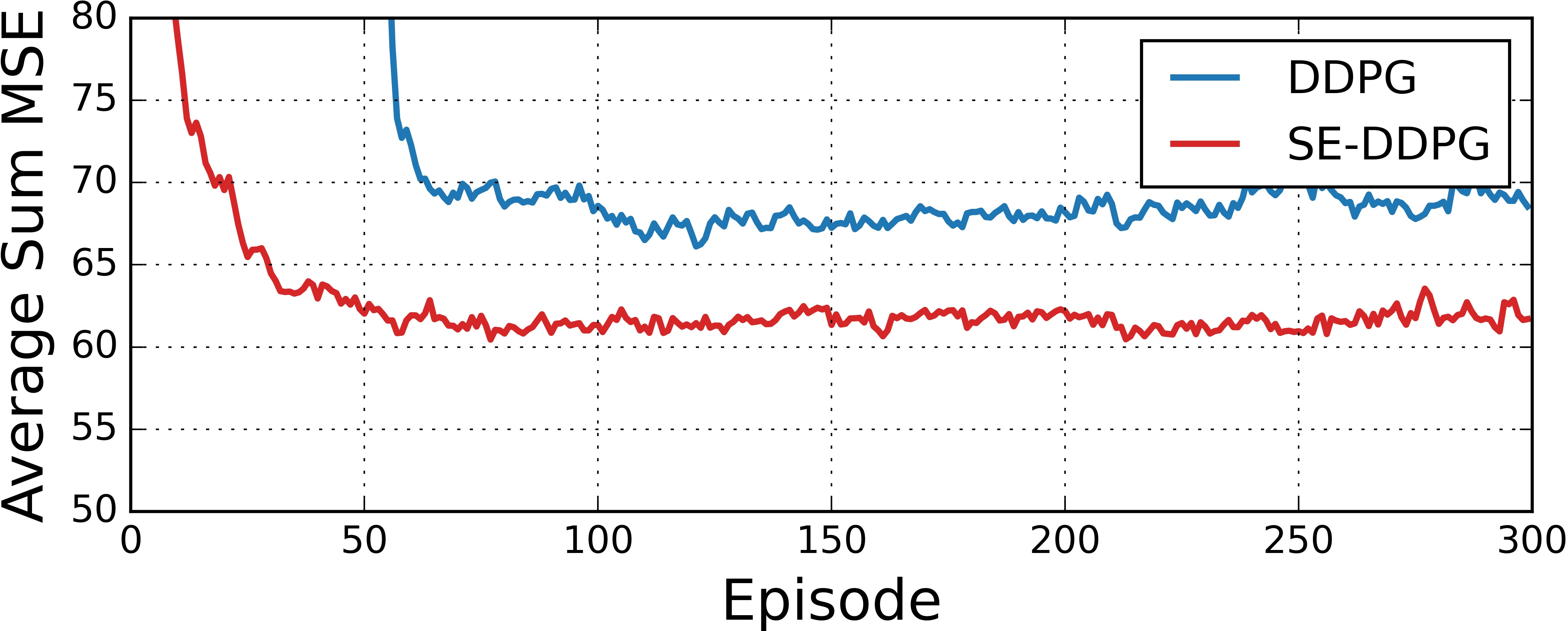}
        \vspace{-0.6cm}
        \captionof{figure}{Average sum MSE of all processes during training with $N\!=10, M\!=5$.}
        \label{fig:training curve DDPG}
        \vspace{-4cm}
    \end{minipage}
    \vspace{-0.7cm}
\end{table}
Our numerical experiments run on a computing
platform with an Intel Core i5 9400F CPU @ 2.9 GHz with 16GB RAM and an NVIDIA RTX 2070 GPU. For the remote estimation system, we set the dimensions of the process state and the measurement as $l_{n} = 2$ and $e_{n} = 1$, respectively. The system matrices $\{\mathbf{A}_{n}\}$ are randomly generated with the spectral radius within the range of $(1, 1.4)$. The entries of $\mathbf{C}_{n}$ are drawn uniformly from the range $(0, 1)$, $\mathbf{W}_n$ and $\mathbf{V}_n$ are identity matrices. The fading channel state is quantized into $\bar{h}=5$ levels, and the corresponding packet drop probabilities are $0.2, 0.15, 0.1, 0.05,$ and $0.01$. The distributions of the channel states of each sensor-channel pair $(n,m)$, i.e., $q^{(n,m)}_1,\dots,q^{(n,m)}_{\bar{h}}$ are generated randomly.

During the training, we use the ADAM optimizer for calculating the gradient and reset the environment for each episode with $T=500$ steps. The episode numbers for the loose SE action, the tight SE action, and the conventional DQN stages are 50, 100, and 150, respectively. The settings of the hyperparameters for Algorithm~\ref{alg:SE-DQN} and Algorithm~\ref{alg:SE-DDPG} are summarized in Table~\ref{tab:Setup}. 


\subsection{Performance Comparison}
Fig.~\ref{fig:training curve DQN} and Fig.~\ref{fig:training curve DDPG} illustrate the average sum MSE of all processes during the training achieved by the SE-DRL algorithms and the benchmarks under some system settings. Fig.~\ref{fig:training curve DQN} shows that the SE-DQN saves about $50\%$ training episodes for the convergence, and also decreases the average sum MSE for $10\%$ than the conventional DQN. 
Fig.~\ref{fig:training curve DDPG} shows that the SE-DDPG saves about $30\%$ training episodes and reduces the average sum MSE by $10\%$, when compared to the conventional DDPG. Also, we see that the conventional DQN and DDPG stages (i.e. the last 150 episodes) in Fig~\ref{fig:training curve DQN} and~\ref{fig:training curve DDPG} cannot improve much of the training performance. This implies that the SE stages have found near optimal policies.
\begin{table}[t]
    \centering
    \begin{minipage}[b]{0.53\textwidth}
    \footnotesize
	\setlength\tabcolsep{3pt}
	\centering
	\caption{Performance Comparison of the SE-DRL and the Benchmarks in terms of the Average Estimation MSE}
	\label{tab:test result}
	\vspace{-0.3cm}
	\begin{tabular}{c|c|c|c|c|c}
		\hline \hline
		\thead{System Scale \\ $(N,M)$} &  \thead{Param. \\ setup} & DQN & \textbf{SE-DQN} & DDPG & \textbf{SE-DDPG}\\
		\hline
		\rowcolor[HTML]{EFEFEF}
		$(6, 3)$ & 1 & 52.4121 & \textbf{47.6766} & 48.4075 & \textbf{47.0594} \\
		& 2 & 67.4247 & \textbf{49.8476} & 53.3675 & \textbf{44.7423} \\
		\rowcolor[HTML]{EFEFEF}
		& 3 & 84.1721 & \textbf{59.9127} & 56.9504 & \textbf{55.2409} \\
		& 4 & 79.5902 & \textbf{65.1640} & 64.4313 & \textbf{58.5534} \\
		\hline
		
		\rowcolor[HTML]{EFEFEF}
		$(6, 2)$ & 5 & 55.7092 & \textbf{50.5983}  & 51.3488 & \textbf{47.1210} \\
		  & 6 & $-$ & \textbf{80.2522} & 77.4124 & \textbf{72.4863} \\
		  \rowcolor[HTML]{EFEFEF}
		  & 7 & $-$ & \textbf{78.7715} & 85.3182 & \textbf{75.8465}\\
		  & 8 & $-$ & \textbf{58.6024} & 61.4121 & \textbf{57.8539}\\
		  \hline
		
		\rowcolor[HTML]{EFEFEF}
		$(10, 5)$ & 9 & $-$ & $-$ & 68.0247 &  \textbf{62.4727} \\
		& 10 & $-$ & $-$ & 89.6290 & \textbf{78.4138} \\
		\rowcolor[HTML]{EFEFEF}
		& 11 & $-$ & $-$ & 90.2812 & \textbf{81.1148} \\
		& 12 & $-$ & $-$ & $-$ & \textbf{147.2844} \\
		\hline
		
		\rowcolor[HTML]{EFEFEF}
		$(20, 10)$ & 13 & $-$ & $-$ & 181.4135 & \textbf{159.4321} \\
		& 14 & $-$ & $-$ & 163.1257 & \textbf{142.1850} \\
		\rowcolor[HTML]{EFEFEF}
		& 15 & $-$ & $-$ & 173.0940 & \textbf{144.8166} \\
		& 16 & $-$ & $-$ & 215.2231 & \textbf{160.2304} \\
		\hline \hline
	\end{tabular}
    \end{minipage}
    \begin{minipage}[b]{0.46\textwidth}
        \centering
        \includegraphics[width=1\textwidth]{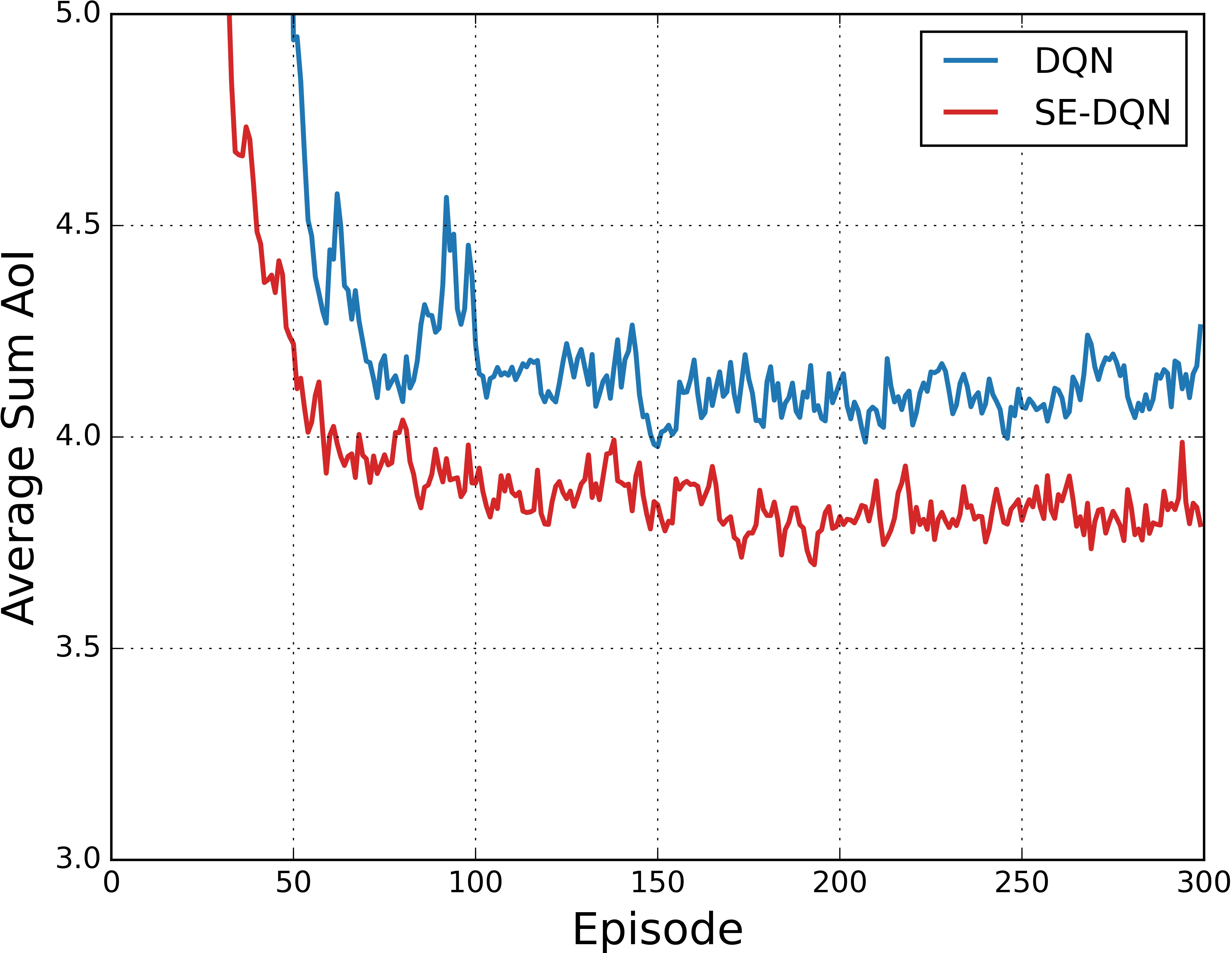}
        \vspace{-0.6cm}
        \captionof{figure}{Average sum AoI of all processes during the training with $N=6, M=3$.}
        \label{fig:training curve AoI}
        \vspace{-3.2cm}
    \end{minipage}
    \vspace{-0.7cm}
\end{table}

In Table~\ref{tab:test result}, we test the performance of well-trained SE-DRL algorithms for different numbers of sensors and channels, and different settings of system parameters, i.e.,  $\mathbf{A}_{n}$, $\mathbf{C}_{n}$, and $q^{(n,m)}_1,\dots,q^{(n,m)}_{\bar{h}}$, based on 10000-step simulations. 
We see that for 6-sensor-3-channel systems, the SE-DQN reduces the average MSE by $10\%$ to $25\%$ over DQN, while both DDPG and SE-DDPG achieve similar performance as the SE-DQN. This suggests that the SE-DQN is almost optimal and that the DDPG and the SE-DDPG cannot improve further. In 6-sensor-2-channel systems, the DQN cannot converge in most experiments but the SE-DQN and SE-DDPG can still perform better than DDPG.
In 10-sensor-5-channel systems, neither the DQN nor the SE-DQN can converge, and the SE-DDPG achieves a $10\%$ MSE reduction over the DDPG. In particular, the SE-DDPG is the only converged algorithm in Experiment 12.
In 20-sensor-10-channel systems, we see that the SE-DDPG can reduce the average MSE by $15\%$ to $25\%$. Therefore, the performance improvement of the SE-DDPG appears significant for large systems.

In addition, to demonstrate that the SE-DRL algorithms can also be applied to other scenarios, we consider a multi-sensor scheduling problem for minimizing the overall expected sum AoI as mentioned in Remark~\ref{remark:counterexample}. In Fig~\ref{fig:training curve AoI}, we see that the SE-DQN can effectively reduce the average sum AoI by $10\%$ compared with the DQN.

We also test the effectiveness of the channel-state threshold property and the loose SE action selection stage in the SE-DRL algorithms.
Fig.~\ref{fig:training curve action} illustrates that without using the channel-threshold property in the SE-action selection process, although the convergence speed is not affected much, such a training method does not converge to a near-optimal policy and leads to a noticeably higher estimation MSE than the original SE-DDPG.
In Fig.~\ref{fig:training curve pretrain}, we see that if one disables the loose SE-action selection stage, the training convergence time doubles compared with the original SE-DDPG, though both converged algorithms provide similar performance.
Thus, both the channel-threshold-enabled action selection and the loose action selection stage are critical to guarantee the performance of the SE-DRL algorithms.

\begin{figure}[t]
    \centering
    \begin{minipage}[t]{0.46\textwidth}
        \centering
        \includegraphics[width=1\textwidth]{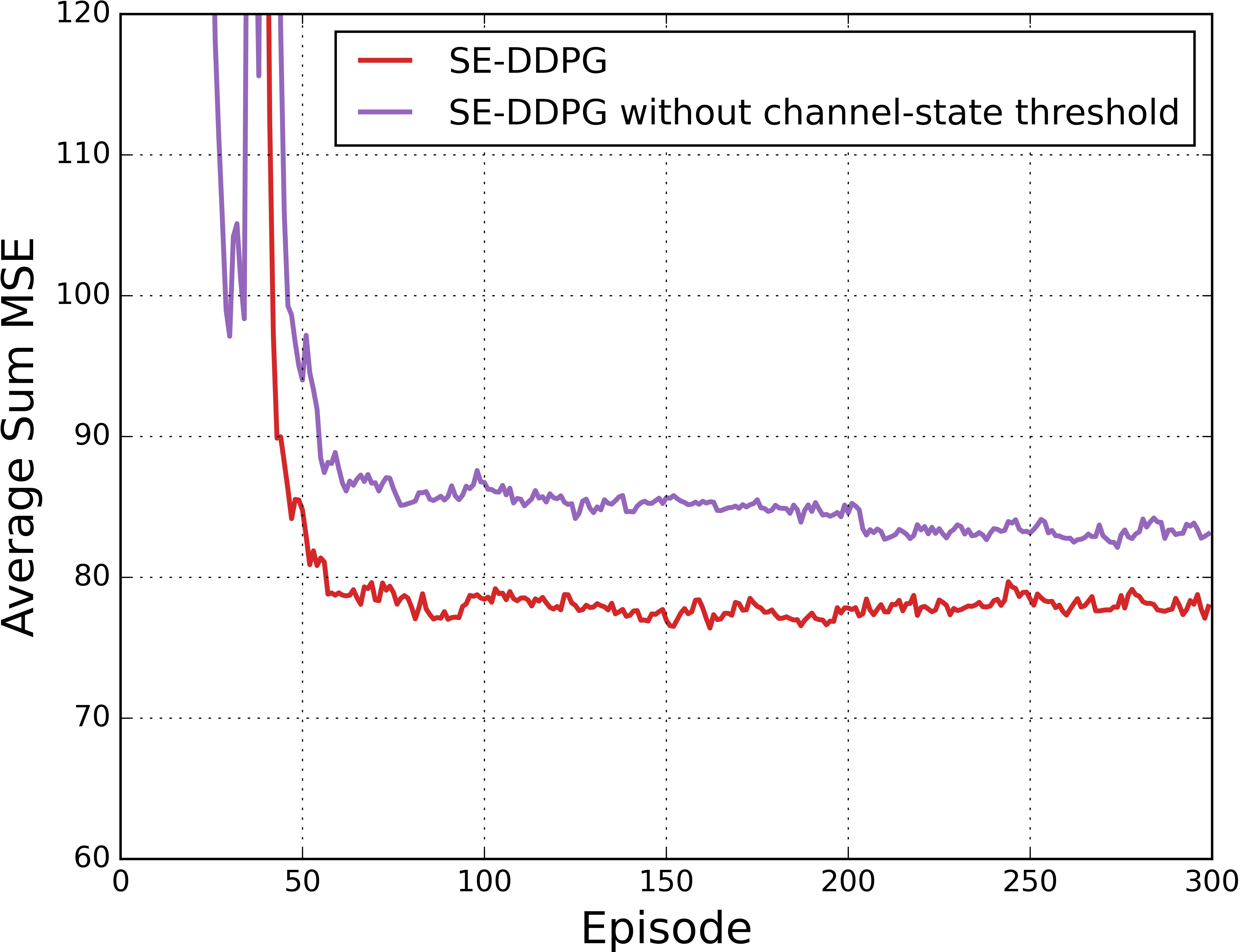}
        \vspace{-0.6cm}
        \caption{Training performance comparison between SE-DDPG with and without using the channel-state threshold property, where $N\!=10, M\!=5$.}
        \label{fig:training curve action}
    \end{minipage}
    \hspace{0.3cm}
    \begin{minipage}[t]{0.46\textwidth}
        \centering
        \includegraphics[width=1\textwidth]{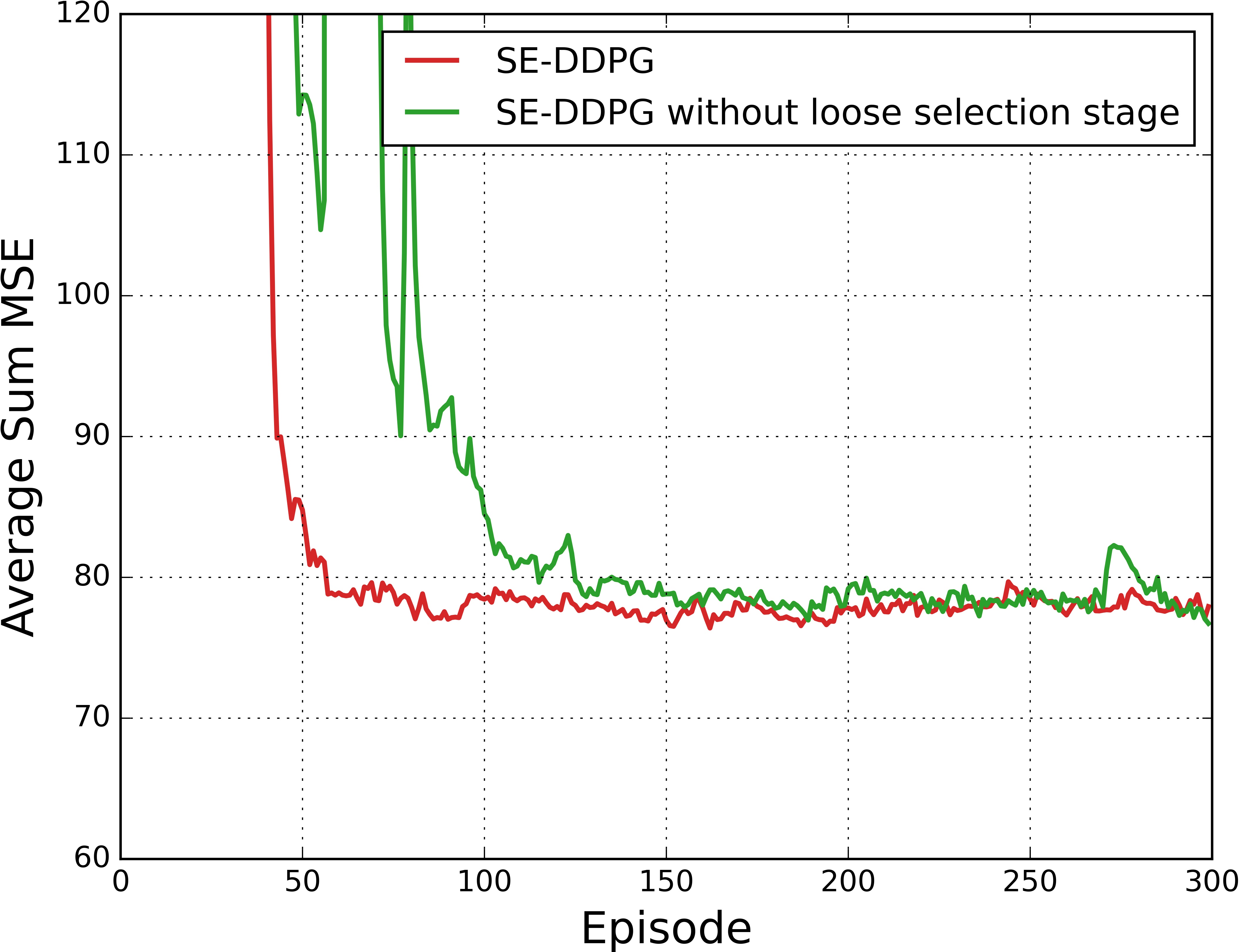}
        \vspace{-0.6cm}
        \caption{Training performance comparison between SE-DDPG with and without the loose SE-action selection stage, where $N\!=10, M\!=5$.}
        \label{fig:training curve pretrain}
    \end{minipage}
    \vspace{-0.7cm}
\end{figure}

\section{Conclusion} \label{sec: conclusion}
In this paper, we have proved that the optimal scheduling policy of the remote estimation system satisfies a number of threshold properties. Then, based on these theoretical guidelines, we have developed the structure-enhanced deep reinforcement learning (SE-DRL) algorithms for solving the optimal scheduling problem. In particular, we have proposed a novel action selection method and a new loss function to improve training efficiency. Our numerical results have illustrated that SE-DRL algorithms can save training time by 50\% and reduce the estimation mean-square error (MSE) by 10\% to 25\%, when compared with benchmark DRL algorithms. In addition, the results also show that SE-DRL can effectively solve a range of optimal scheduling problems, not limited to remote state estimation systems. For future work, we will investigate other structural properties of optimal scheduling and resource allocation policies for wireless communications and networked systems, and then use these properties to develop tailored DRL algorithms.



\appendices
\section{Proof of Lemma~\ref{lemma:monotone}}\label{proof: monotonicity}
Similar to the Q-value function, we define a W-function for the value function, $W(\mathbf{s}, \mathbf{a}, V^{\tilde{t}}): \mathcal{S} \times \mathcal{A} \times \mathcal{V} \to \mathbb{R}$
\begin{equation}
    W(\mathbf{s}, \mathbf{a}, V^{\tilde{t}}) = r(\mathbf{s}) + \gamma \sum_{\mathbf{s}^{+}} \operatorname{Pr}(\mathbf{s}^{+}|\mathbf{s}, \mathbf{a}) V^{\tilde{t}} (\mathbf{s}^{+}). \label{eq:W-function}
\end{equation}
According to Lemma~\ref{lemma:converge of V*}, the converge of the optimal value function, $V^{*}(\mathbf{s})$, is independent of the initial value function $V^{0}(\mathbf{s})$. Therefore, given states $\mathbf{s} = (\bm{\tau}, \mathbf{H})$ and $ \mathbf{s}' = (\bm{\tau}'_{(i)}, \mathbf{H})$, where $\tau'_{i} \geq \tau_{i}$, we can assume that $V^{0}(\mathbf{s})$ is monotonic, i.e.
\begin{equation}\label{eq:monotonicity,0}
    V^{0}\left(\mathbf{s}'\right) \leq V^{0}\left(\mathbf{s}\right).
\end{equation}
To prove Lemma~\ref{lemma:monotone} based on~\eqref{eq:monotonicity,0} and Lemma~\ref{lemma:converge of V*}, it is sufficient to show that the value function~$V^{1}(\mathbf{s})$ is monotonic, i.e.
\begin{equation}\label{eq:monotonicity,1}
    V^{1}\left(\mathbf{s}'\right) \leq V^{1}\left(\mathbf{s}\right),
\end{equation}
and then the monotonicity of the value function can be preserved by the Bellman operator $\mathsf{B}[\cdot]$ from $V^{0}(\mathbf{s})$ to $V^{*}(\mathbf{s})$.
During the value iteration, we define the optimal action and the optimal policy at $\tilde{t}$th iteration as 
\begin{equation}
    \mathbf{a}^{\tilde{t}} \triangleq \pi^{\tilde{t}}(\mathbf{s}) = r(\mathbf{s}) + \gamma \mathop{\arg\max}_{\mathbf{a}\in \mathcal{A}}  \left[ \sum_{\mathbf{s}^{+}} \operatorname{Pr}(\mathbf{s}^{+}|\mathbf{s}, \mathbf{a})V^{\tilde{t}-1}(\mathbf{s}^{+})\right]. \label{eq:optimal action V_t}
\end{equation}
From~\eqref{eq: bellman operator} and~\eqref{eq:optimal action V_t}, it directly follows the relationship between the value function and the W-functions:
\begin{equation}\label{ineq:V&W}
    V^{\tilde{t}+1} (\mathbf{s}) = W(\mathbf{s}, \mathbf{a}^{\tilde{t}}, V^{\tilde{t}}) \geq W(\mathbf{s}, \mathbf{a}, V^{\tilde{t}}). 
\end{equation}

To prove \eqref{eq:monotonicity,1}, we use~\eqref{ineq:V&W} and the optimal action $\mathbf{a}^{1} = \pi^{1}\left( \mathbf{s}' \right)$ to derive the following inequality,
\begin{align}\label{eq:app A V->W}
    V^{1}\left (\mathbf{s}'\right) - V^{1}\left(\mathbf{s}\right) \leq W(\mathbf{s}', \mathbf{a}^{1}, V^{0}) - W(\mathbf{s}, \mathbf{a}^{1}, V^{0}).
\end{align}
By using~\eqref{eq: transition p} and~\eqref{eq: transition tau},~\eqref{eq:app A V->W} can be derived that
\begin{align}
    & W(\mathbf{s}', \mathbf{a}^{1}, V^{0}) - W(\mathbf{s}, \mathbf{a}^{1}, V^{0}) \\
    & = \left[r\left(\bm{\tau}'\right) -  r\left(\bm{\tau}\right) \right] + \gamma \Bigg[ \sum_{\mathbf{H}^{+}} \sum_{{\bm{\tau}_{\backslash i}}^{+}} \sum_{{\tau_{i}'}^{+}} \operatorname{Pr} (\mathbf{H}^{+}) \operatorname{Pr} ({\bm{\tau}_{\backslash i}}^{+}|\bm{\tau}_{\backslash i}, \mathbf{H}_{\backslash i}, \mathbf{a}^{1}_{\backslash i}) \operatorname{Pr} ({\tau'_{i}}^{+}|{\tau'_{i}}, \mathbf{h}_{i}, a^{1}_{i}) V^{0}({\mathbf{s}'}^{+}) \\
    \label{eq: W i -i}
    & \quad - \sum_{\mathbf{H}^{+}} \sum_{{\bm{\tau}_{\backslash i}}^{+}} \sum_{{\tau_{i}}^{+}} \operatorname{Pr} ({\bm{\tau}_{\backslash i}}^{+}|\bm{\tau}_{\backslash i}, \mathbf{H}_{\backslash i}, \mathbf{a}^{1}_{\backslash i}) \operatorname{Pr} ({\tau_{i}}^{+}|{\tau_{i}}, \mathbf{h}_{i}, a^{1}_{i}) V^{0}(\mathbf{s}^{+})\Bigg],
\end{align}
which implies that the states except for $\tau_{i}$ are constant. Therefore, the states can be written as $\mathbf{s} = \tau_{i}$ and $ \mathbf{s}' = \tau'_{i}$ for notation simplicity. Then, using~\eqref{eq:monotonicity,0} and $\operatorname{Pr} ({\tau'_{i}}^{+}={\tau'_{i}}+1|{\tau'_{i}}, \mathbf{h}_{i}, a^{1}_{i}) = \operatorname{Pr} ({\tau_{i}}^{+}={\tau_{i}}+1|{\tau_{i}}, \mathbf{h}_{i}, a^{1}_{i})$, we have 
\begin{align}\label{eq:V(tau')<V(tau)}
    \operatorname{Pr} ({\tau'_{i}}^{+}={\tau'_{i}}+1|{\tau'_{i}}, \mathbf{h}_{i}, a^{1}_{i}) V^{0}({\tau'_{i}}+1) \leq \operatorname{Pr} ({\tau_{i}}^{+}={\tau_{i}}+1|{\tau_{i}}, \mathbf{h}_{i}, a^{1}_{i}) V^{0}({\tau_{i}}+1).
\end{align}
Similarly, we use $\operatorname{Pr} ({\tau'_{i}}^{+}=1|{\tau'_{i}}, \mathbf{h}_{i}, a^{1}_{i}) = \operatorname{Pr} ({\tau_{i}}^{+}=1|{\tau_{i}}, \mathbf{h}_{i}, a^{1}_{i})$ to derive that
\begin{align}\label{eq:V(tau)=V(tau)}
    \operatorname{Pr} ({\tau'_{i}}^{+}=1|{\tau'_{i}}, \mathbf{h}_{i}, a^{1}_{i}) V^{0}(1) = \operatorname{Pr} ({\tau_{i}}^{+}=1|{\tau_{i}}, \mathbf{h}_{i}, a^{1}_{i}) V^{0}(1).
\end{align}
Based on~\eqref{eq:monotonicity,0},~\eqref{eq:V(tau')<V(tau)},~\eqref{eq:V(tau)=V(tau)}, and $r\left(\tau'_{i}\right) <  r\left(\tau_{i}\right)$, the following inequality can be derived from~\eqref{eq: W i -i}
\begin{align}\label{eq:r(tau')-r(tau) + V(tau')-V(tau) < 0}
    \!\!\!\! \left[r\left(\tau'_{i}\right) -  r\left(\tau_{i}\right) \right] 
    + \gamma \bigg[ \sum_{{\tau_{i}'}^{+}} \operatorname{Pr} ({\tau'_{i}}^{+}|{\tau'_{i}}, \mathbf{h}_{i}, a^{1}_{i}) V^{0}({\tau'_{i}}^{+}) - \sum_{{\tau_{i}}^{+}} \operatorname{Pr} ({\tau_{i}}^{+}|{\tau_{i}}, \mathbf{h}_{i}, a^{1}_{i}) V^{0}({\tau_{i}}^{+}) \bigg] \leq 0.
\end{align}
From~\eqref{eq:app A V->W},~\eqref{eq: W i -i}, and~\eqref{eq:r(tau')-r(tau) + V(tau')-V(tau) < 0}, we have $V^{1}\left (\mathbf{s}'\right) \leq V^{1}\left(\mathbf{s}\right)$.

Thus, the monotonicity of the value function $V^{0}(\mathbf{s})$ propagates through the Bellman operator~$\mathsf{B}[\cdot]$ to the optimal value function $V^{*}(\mathbf{s})$.

\section{Proof of Lemma~\ref{lemma:probabilistic sup}} \label{proof: sup and sub, N=2}
In this proof, we also need the submodularity of the value function. Therefore, similar to the proof of Lemma~\ref{lemma:monotone}, we assume that the initial value function $V^{0}(\mathbf{s})$ is submodular and probabilistic supermodular. Thus, given states $\mathbf{s} = (\bm{\tau}, \mathbf{H})$ and $ \mathbf{s}^{\circ} = (\bm{\tau}^{\circ}, \mathbf{H})$ where $\bm{\tau}^\circ = (\tau''_i, \tau'_{j})$ with $\tau''_{i} \leq \tau_{i}$ and $\tau'_{j} \geq \tau_{j}$, we obtain that
\begin{subequations}
    \begin{align}
        V^{0}\left(\mathbf{s} \vee \mathbf{s}^{\circ}\right) + V^{0}\left(\mathbf{s} \wedge \mathbf{s}^{\circ}\right) & \leq V^{0}\left(\mathbf{s}\right) + V^{0}\left(\mathbf{s}^{\circ}\right), \label{eq:sub_0} \\
        p_{j,1} V^{0}\left(\mathbf{s} \wedge \mathbf{s}^{\circ}\right) + (p_{j,1} - p_{i,1}) V^{0}\left(\mathbf{s} \vee \mathbf{s}^{\circ}\right) & \geq p_{j,1} V^{0}\left(\mathbf{s}\right) + (p_{j,1} - p_{i,1}) V^{0}\left(\mathbf{s}^{\circ}\right).\label{eq:sup_0}  
    \end{align}
\end{subequations}
To prove Lemma~\ref{lemma:probabilistic sup} based on~\eqref{eq:sub_0},~\eqref{eq:sup_0} and Lemma~\ref{lemma:converge of V*}, it is sufficient to show that the value function $V^{1}(\mathbf{s})$ is submodular and probabilistic supermodular, i.e.
\begin{subequations}
    \begin{align}
    V^{1}\left(\mathbf{s} \vee \mathbf{s}^{\circ}\right) + V^{1}\left(\mathbf{s} \wedge \mathbf{s}^{\circ}\right) & \leq V^{1}\left(\mathbf{s}\right) + V^{1}\left(\mathbf{s}^{\circ}\right), \label{eq:sup&sub_1,a} \\
        p_{j,1} V^{1}\left(\mathbf{s} \wedge \mathbf{s}^{\circ}\right) + (p_{j,1} - p_{i,1}) V^{1}\left(\mathbf{s} \vee \mathbf{s}^{\circ}\right) & \geq p_{j,1} V^{1}\left(\mathbf{s}\right) + (p_{j,1} - p_{i,1}) V^{1}\left(\mathbf{s}^{\circ}\right), \label{eq:sup&sub_1,b} 
    \end{align}
\end{subequations}
and then these properties of the value function can be preserved by the Bellman operator $\mathsf{B}[\cdot]$ from $V^{0}(\mathbf{s})$ to $V^{*}(\mathbf{s})$.

Similar to the proof of Lemma~\ref{lemma:monotone}, we write the states as $\mathbf{s} = \bm{\tau}$ and $\mathbf{s}^{\circ} = \bm{\tau}^{\circ}$, as the channel states are constant. Then, we have $\bm{\tau} \vee \bm{\tau}^{\circ} = (\tau_{i}, \tau'_{j})$ and $\bm{\tau} \wedge \bm{\tau}^{\circ} = (\tau''_{i}, \tau_{j})$.
In terms of the policy, since we only consider the system with a single channel in this proof, we write the optimal policy as ${\pi^{\tilde{t}}}(\bm{\tau}) = i$ to represent that sensor $i$ is scheduled for the state $\bm{\tau}$, i.e. $a^{\tilde{t}}_{i}=1$. 

In the following, we prove~\eqref{eq:sup&sub_1,a} by cases \textbf{(a)} and \textbf{(b)}, and~\eqref{eq:sup&sub_1,b} by cases \textbf{(a')} and \textbf{(b')} with different packet success rates.
\begin{itemize}
    \item[\textbf{(a)}] If $p_{j,1} \leq p_{i,1}$, then there are four cases with different optimal actions of the states:

    (a.1) ${\pi^{1}} (\bm{\tau} \vee \bm{\tau}^{\circ}) = {\pi^{1}} (\bm{\tau} \wedge \bm{\tau}^{\circ}) = i$,
    (a.2) ${\pi^{1}} (\bm{\tau} \vee \bm{\tau}^{\circ}) = {\pi^{1}} (\bm{\tau} \wedge \bm{\tau}^{\circ}) = j$,
    (a.3) ${\pi^{1}} (\bm{\tau} \vee \bm{\tau}^{\circ}) = i$ and ${\pi^{1}} (\bm{\tau} \wedge \bm{\tau}^{\circ}) = j$,
    (a.4) ${\pi^{1}} (\bm{\tau} \vee \bm{\tau}^{\circ}) = j$ and ${\pi^{1}} (\bm{\tau} \wedge \bm{\tau}^{\circ}) = i$.

    \item[(a.1)] If ${\pi^{1}} (\bm{\tau} \vee \bm{\tau}^{\circ}) = {\pi^{1}} (\bm{\tau} \wedge \bm{\tau}^{\circ}) = i$, then based on~\eqref{ineq:V&W}, we have
    \begin{align}
        & V^{1}(\bm{\tau} \vee \bm{\tau}^{\circ}) + V^{1}(\bm{\tau} \wedge \bm{\tau}^{\circ}) - V^{1}(\bm{\tau}) - V^{1}(\bm{\tau}^{\circ}) \\
        \label{eq:appb case a.1}
        & \leq W(\bm{\tau} \vee \bm{\tau}^{\circ}, i, V^{0}) + W(\bm{\tau} \wedge \bm{\tau}^{\circ}, i, V^{0}) - W(\bm{\tau}, i, V^{0}) - W(\bm{\tau}^{\circ}, i, V^{0}).
    \end{align}
    Using~\eqref{eq:W-function} and $r(\bm{\tau} \vee \bm{\tau}^{\circ}) + r(\bm{\tau} \wedge \bm{\tau}^{\circ}) = r(\bm{\tau}) + r(\bm{\tau}^{\circ})$,~\eqref{eq:appb case a.1} can be derived that
    \begin{align}
        & \gamma p_{i,1} \left[ V^{0}(1, \tau'_{j} + 1) + V^{0}(1, \tau_{j} + 1) - V^{0}(1, \tau_{j} + 1) - V^{0}(1, \tau'_{j} + 1) \right] \\
        & + \gamma (1\!-\!p_{i,1}) \big[ V^{0}(\tau_{i} \!+\! 1, \tau'_{j} \!+\! 1) \!+\! V^{0}(\tau''_i \!+\! 1, \tau_{j} \!+\! 1) \!-\! V^{0}(\tau_{i} \!+\! 1, \tau_{j} \!+\! 1) \!-\! V^{0}(\tau''_i \!+\! 1, \tau'_{j} \!+\! 1) \big] \\
        \label{eq:appb case a.1 <0}
        & \leq 0,
    \end{align}
    where the inequality is based on~\eqref{eq:sub_0}. From~\eqref{eq:appb case a.1} and~\eqref{eq:appb case a.1 <0}, we can derive~\eqref{eq:sup&sub_1,a}.
    \item[(a.2)] If ${\pi^{1}} (\bm{\tau} \vee \bm{\tau}^{\circ}) = {\pi^{1}} (\bm{\tau} \wedge \bm{\tau}^{\circ}) = j$, the proof is similar to the case (a.1) by showing $W(\bm{\tau} \vee \bm{\tau}^{\circ}, j, V^{0}) + W(\bm{\tau} \wedge \bm{\tau}^{\circ}, j, V^{0}) - W(\bm{\tau}, j, V^{0}) - W(\bm{\tau}^{\circ}, j, V^{0}) \leq 0$.
    
    \item[(a.3)] If ${\pi^{1}} (\bm{\tau} \vee \bm{\tau}^{\circ}) = i$ and ${\pi^{1}} (\bm{\tau} \wedge \bm{\tau}^{\circ}) = j$, then
    \begin{align*}
        & \hspace{-0.4cm} V^{1}(\bm{\tau} \vee \bm{\tau}^{\circ}) + V^{1}(\bm{\tau} \wedge \bm{\tau}^{\circ}) - V^{1}(\bm{\tau}) - V^{1}(\bm{\tau}^{\circ}) \\
        & \hspace{-0.4cm} \!\!\leq\! W(\bm{\tau} \vee \bm{\tau}^{\circ}, i, V^{0}) + W(\bm{\tau} \wedge \bm{\tau}^{\circ}, j, V^{0}) - W(\bm{\tau}, i, V^{0}) - W(\bm{\tau}^{\circ}, j, V^{0})\\
        & \hspace{-0.4cm} \!\!=\! \gamma \big[p_{i,1} V^{0}(1, \tau'_{j} \!+\! 1) \!+\! (1\!-\!p_{i,1}) V^{0}(\tau_{i} \!+\! 1, \tau'_{j} \!+\! 1) \!+\! p_{j,1} V^{0}(\tau''_i \!+\! 1, 1) \!+\! (1\!-\!p_{j,1}) V^{0}(\tau''_i \!+\! 1, \tau_{j} \!+\! 1) \\
        & \hspace{-0.4cm} \quad \!-\! p_{i,1} V^{0}(1, \tau_{j} \!+\! 1) \!-\! (1\!-\!p_{i,1}) V^{0}(\tau_{i} \!+\! 1,\! \tau_{j} \!+\! 1) \!-\! p_{j,1} V^{0}(\tau''_i \!+\! 1, 1) \!-\! (1\!-\!p_{j,1}) V^{0}(\tau''_i \!+\! 1,\! \tau'_{j} \!+\! 1)\big] \\
        & \hspace{-0.4cm} \!\!=\! \gamma \big[p_{i,1} V^{0}(1,\! \tau'_{j} \!+\!\! 1) \!+\! (p_{i,1} \!-\! p_{j,1}\!) V^{0}(\tau''_i \!\!+\!\! 1,\! \tau_{j} \!+\!\! 1) \!-\! p_{i,1} V^{0}(1,\! \tau_{j} \!+\!\! 1) \!-\! (p_{i,1} \!-\! p_{j,1}) V^{0}(\tau''_i \!\!+\!\! 1,\! \tau'_{j} \!+\!\! 1) \big] \\
        & \hspace{-0.4cm} \quad + \gamma (1\!-\!p_{i,1}) \big[ V^{0}\big(\tau_{i} \!+\! 1, \tau'_{j} \!+\! 1) \!+\! V^{0}(\tau''_i \!+\! 1, \tau_{j} \!+\! 1)  \!-\! V^{0}(\tau_{i} \!+\! 1, \tau_{j} \!+\! 1) \!-\! V^{0}(\tau''_i \!+\! 1,  \tau'_{j} \!+\! 1) \big] \\
        & \hspace{-0.4cm} \!\leq\! 0,
    \end{align*}
    where the last inequality is based on \eqref{eq:sub_0} and \eqref{eq:sup_0}.
    \item[(a.4)] If ${\pi^{1}} (\bm{\tau} \vee \bm{\tau}^{\circ}) = j$ and ${\pi^{1}} (\bm{\tau} \wedge \bm{\tau}^{\circ}) = i$, then
    \begin{align}
        & \hspace{-0.7cm} V^{1}(\bm{\tau} \vee \bm{\tau}^{\circ}) + V^{1}(\bm{\tau} \wedge \bm{\tau}^{\circ}) - V^{1}(\bm{\tau}) - V^{1}(\bm{\tau}^{\circ}) \\
        & \hspace{-0.7cm} \!\leq\! W(\bm{\tau} \vee \bm{\tau}^{\circ}, j, V^{0}) + W(\bm{\tau} \wedge \bm{\tau}^{\circ}, i, V^{0}) - W(\bm{\tau}, i, V^{0}) - W(\bm{\tau}^{\circ}, j, V^{0})\\
        & \hspace{-0.7cm} \!=\! \gamma p_{j,1} [V^{0}(\tau_{i} \!+\! 1, 1) \!-\! V^{0}(\tau''_i \!+\! 1, 1)] \!+\! \gamma (1\!-\!p_{j,1}) [V^{0}(\tau_{i} \!+\! 1, \tau'_{j} \!+\! 1) \!-\! V^{0}(\tau''_i \!+\! 1, \tau'_{j} \!+\! 1)] \\
        & \hspace{-0.7cm} \quad \!+\! \gamma (1-p_{i,1}) [V^{0}(\tau''_{i} + 1, \tau_{j} + 1) - V^{0}(\tau_{i} + 1, \tau_{j} + 1)] \\
        \label{eq:appb case a.4}
        & \hspace{-0.7cm} \!\leq\! \gamma (1\!-\!p_{i,1}) [V^{0}(\tau_{i} \!+\! 1, \tau'_{j} \!+\! 1) \!+\! V^{0}(\tau_{i} \!+\! 1, \tau_{j} \!+\! 1) \!-\! V^{0}(\tau''_i \!+\! 1, \tau'_{j} \!+\! 1) \!-\! V^{0}(\tau_{i} \!+\! 1, \tau_{j} \!+\! 1)],
    \end{align}
    where the last inequality is based on $p_{j,1} \leq p_{i,1}, V^{0}(\tau_{i} \!+\! 1, 1) \leq V^{0}(\tau''_i \!+\! 1, 1)$ and $V^{0}(\tau_{i} \!+\! 1, \tau'_{j} \!+\! 1) \leq V^{0}(\tau''_i \!+\! 1, \tau'_{j} \!+\! 1)$ from the proof of Lemma~\ref{lemma:monotone}.
    From~\eqref{eq:sub_0} and~\eqref{eq:appb case a.4}, we have
    \begin{align}
        V^{1}(\bm{\tau} \vee \bm{\tau}^{\circ}) + V^{1}(\bm{\tau} \wedge \bm{\tau}^{\circ}) - V^{1}(\bm{\tau}) - V^{1}(\bm{\tau}^{\circ}) \leq 0,
    \end{align}
    which is exactly~\eqref{eq:sup&sub_1,a}.
    \item[\textbf{(b)}] If $p_{j,1} \geq p_{i,1}$, the proof is similar to the case \textbf{(a)}.
\end{itemize}

Similarly, we prove~\eqref{eq:sup&sub_1,b} by different cases with different packet success rates.

\begin{itemize}
    \item[\textbf{(a')}] If $p_{j,1} \leq p_{i,1}$, then
    \begin{align}
        & \!\!\!\!p_{j,1} V^{1}(\bm{\tau}) + (p_{j,1} - p_{i,1}) V^{1}(\bm{\tau}^{\circ}) - p_{j,1} V^{1}(\bm{\tau} \wedge \bm{\tau}^{\circ}) - (p_{j,1} - p_{i,1}) V^{1}(\bm{\tau} \vee \bm{\tau}^{\circ}) \\
        & \!\!\!\!\!= \!p_{j,1} \! \left[ V^{1}(\tau_{i} \!+\!\! 1, \tau_{j} \!+\!\! 1) \!-\! V^{1}(\tau''_i \!+\!\! 1, \tau_{j} \!+\!\! 1) \right] \!+\! (p_{i,1} \!-\! p_{j,1}) \!\left[ V^{1}(\tau_{i} \!+\!\! 1, \tau'_{j} \!+\!\! 1) \!-\! V^{1}(\tau''_i \!+\!\! 1, \tau'_{j} \!+\!\! 1) \right], \!\!\\
        & \!\!\!\!\!\leq 0,
    \end{align}
    where the last inequality is based on the proof of Lemma~\ref{lemma:monotone}. 
    
    \item[\textbf{(b')}] If $p_{j,1} \geq p_{i,1}$, then there are four cases with different optimal actions of the states:
    
    (b'.1)\! ${\pi^{1}} (\bm{\tau}) \!=\! {\pi^{1}} (\bm{\tau}^{\circ}) \!=\! i$,
    (b'.2) \!${\pi^{1}} (\bm{\tau}) \!=\! {\pi^{1}} (\bm{\tau}^{\circ}) \!=\! j$,
    (b'.3) \!${\pi^{1}} (\bm{\tau}) \!=\! i$ and ${\pi^{1}} (\bm{\tau}^{\circ}) \!=\! j$,
    (b'.4) \!${\pi^{1}} (\bm{\tau}) \!=\! j$ and ${\pi^{1}} (\bm{\tau}^{\circ}) \!=\! i$.

    \item[(b'.1)] If ${\pi^{1}} (\bm{\tau}) = {\pi^{1}} (\bm{\tau}^{\circ}) = i$, then
    \begin{align}
        & \hspace{-0.5cm} p_{j,1}  V^{1}(\bm{\tau}) + (p_{j,1} - p_{i,1})  V^{1}(\bm{\tau}^{\circ}) - p_{j,1}  V^{1}(\bm{\tau} \wedge \bm{\tau}^{\circ}) - (p_{j,1} - p_{i,1})  V^{1}(\bm{\tau} \vee \bm{\tau}^{\circ}) \\
        & \hspace{-0.5cm} \!\leq\! p_{j,1} W(\bm{\tau}, i, \!V^{0}) \!+\! (p_{j,1} \!-\! p_{i,1}) W(\bm{\tau}^{\circ}\!, i, \!V^{0}) \!-\! p_{j,1} W(\bm{\tau} \!\wedge\! \bm{\tau}^{\circ}\!, i, \!V^{0}) \!-\! (p_{j,1} \!-\! p_{i,1}) W(\bm{\tau} \!\vee\! \bm{\tau}^{\circ}\!, i, \!V^{0}) \! \\
        & \hspace{-0.5cm} \!=\! p_{i,1} \left[ r(\tau_{i}, \tau'_{j}) \!-\! r(\tau''_i, \tau'_{j}) \right] \!+\! \gamma p_{j,1} \left[ p_{i,1} V^{0}(1, \tau_{j} \!+\! 1) \!+\! (1-p_{i,1}) V^{0}(\tau_{i} \!+\! 1, \tau_{j} \!+\! 1) \right] \\
        & \hspace{-0.5cm} \quad + \gamma (p_{j,1} - p_{i,1}) \left[ p_{i,1} V^{0}(1, \tau'_{j} + 1) + (1-p_{i,1}) V^{0}(\tau''_i + 1, \tau'_{j} + 1) \right] \\
        & \hspace{-0.5cm} \quad - \gamma p_{j,1} \left[ p_{i,1} V^{0}(1, \tau_{j} + 1) + (1-p_{i,1}) V^{0}(\tau''_i + 1, \tau_{j} + 1) \right] \\
        \label{eq:appb case b'1 1}
        & \hspace{-0.5cm} \quad - \gamma (p_{j,1} - p_{i,1}) \left[ p_{i,1} V^{0}(1, \tau'_{j} + 1) + (1-p_{i,1}) V^{0}(\tau_{i} + 1, \tau'_{j} + 1) \right].
    \end{align}
    From~\eqref{eq:sup_0}, we have
    \begin{align}\label{eq:appb case b'1 2}
        \!\!\!\!\!\! p_{j,1} V^{0}(1, \tau_{j} \!+\! 1) + (p_{j,1} - p_{i,1}) V^{0}(1, \tau'_{j} + 1) - V^{0}(1, \tau_{j} + 1) - p_{i,1} V^{0}(1, \tau'_{j} + 1) \leq 0
    \end{align}
    and 
    \begin{align}\label{eq:appb case b'1 3}
        & p_{j,1} V^{0}(\tau_{i} + 1, \tau_{j} + 1) + (p_{j,1} - p_{i,1}) V^{0}(\tau''_i + 1, \tau'_{j} + 1) \\
        & - p_{j,1} V^{0}(\tau''_i + 1, \tau_{j} + 1) - (p_{j,1} - p_{i,1}) V^{0}(\tau_{i} + 1, \tau'_{j} + 1) \leq 0.
    \end{align}
    Based on $r(\tau_{i}, \tau'_{j}) \leq r(\tau''_i, \tau'_{j})$,~\eqref{eq:appb case b'1 1},~\eqref{eq:appb case b'1 2}, and~\eqref{eq:appb case b'1 3}, it can be derived that
    \begin{align}
        V^{1}(\bm{\tau}) + (p_{j,1} - p_{i,1})  V^{1}(\bm{\tau}^{\circ}) - p_{j,1}  V^{1}(\bm{\tau} \wedge \bm{\tau}^{\circ}) - (p_{j,1} - p_{i,1})  V^{1}(\bm{\tau} \vee \bm{\tau}^{\circ}) \leq 0,
    \end{align}
    which is exactly~\eqref{eq:sup&sub_1,b}.
    \item[(b'.2)] If ${\pi^{1}} (\bm{\tau}) = {\pi^{1}} (\bm{\tau}^{\circ}) = j$, the proof is similar to the case (b'.1) by showing $p_{j,1} W(\bm{\tau}, j, V^{0}) \!+\! (p_{j,1} \!-\! p_{i,1}) W(\bm{\tau}^{\circ}, j, V^{0}) \!-\! p_{j,1} W(\bm{\tau} \!\wedge\! \bm{\tau}^{\circ}, j, V^{0}) \!-\! (p_{j,1} \!-\! p_{i,1}) W(\bm{\tau} \!\vee\! \bm{\tau}^{\circ}, j, V^{0}) \leq 0$.

    \item[(b'.3)] If ${\pi^{1}} (\bm{\tau}) = i$ and ${\pi^{1}} (\bm{\tau}^{\circ}) = j$, then
    \vspace{-0.5cm}
    \begin{align}
        & \hspace{-0.8cm} p_{j,1}  V^{1}(\bm{\tau}) + (p_{j,1} - p_{i,1})  V^{1}(\bm{\tau}^{\circ}) - p_{j,1}  V^{1}(\bm{\tau} \wedge \bm{\tau}^{\circ}) - (p_{j,1} - p_{i,1})  V^{1}(\bm{\tau} \vee \bm{\tau}^{\circ}) \\
        & \hspace{-0.8cm} \!\leq\! p_{j,1} W(\bm{\tau}, i, \!V^{0}) \!+\! (p_{j,1} \!-\! p_{i,1}) W(\bm{\tau}^{\circ}\!, j, \!V^{0}) \!-\! p_{j,1} W(\bm{\tau} \!\wedge\! \bm{\tau}^{\circ}\!, i, \!V^{0}) \!-\! (p_{j,1} \!-\! p_{i,1}) W(\bm{\tau} \!\vee\! \bm{\tau}^{\circ}\!, j, \!V^{0})  \\
        & \hspace{-0.8cm} \!\leq \!\! \gamma p_{j,1} \!\!\left[ p_{j,1} \!V^{0}(\tau''_i \!\!+\!\! 1,\! 1) \!+\! (p_{j,1}\!\!-\!p_{i,1}) V^{0}(\tau_i \!+\!\! 1, \!\tau'_j \!+\!\! 1) \!-\! p_{j,1} \!V^{0}(\tau_i \!+\!\! 1,\! 1) \!-\! (p_{j,1}\!\!-\!p_{i,1}) V^{0}(\tau''_i \!\!+\!\! 1,\! \tau'_j \!+\!\! 1)\right] \!\!\! \\
        & \hspace{-0.8cm} \quad - \gamma [(p_{j,1} - p_{i,1}) V^{0}(\tau_i + 1, \tau'_j + 1) -(p_{j,1} - p_{i,1}) V^{0}(\tau''_i + 1, \tau'_j + 1)]  \\
        & \hspace{-0.8cm} \quad + \gamma [p_{j,1}  V^{0}(\tau_i + 1, \tau_j + 1) - p_{j,1} V^{0}(\tau''_i + 1, \tau_j + 1) ] \\
        \label{eq:appb case b'3 1}
        & \hspace{-0.8cm} \quad - \gamma p_{i,1} p_{j,1} \left[  V^{0}(\tau''_i + 1, 1) + V^{0}(\tau_i + 1, \tau_j + 1) - V^{0}(\tau_i + 1, 1) - V^{0}(\tau''_i + 1, \tau_j + 1)\right],
    \end{align}
    where the inequality is based on $r(\tau_{i}, \tau'_{j}) \leq r(\tau''_i, \tau'_{j})$. Then, from~\eqref{eq:sup_0}, we have
    \begin{align}\label{eq:appb case b'3 sup}
        \hspace{-1.3cm} p_{j,1} V^{0}(\tau''_i \!+\!\! 1, \!1) \!+\! (p_{j,1}\!-\!p_{i,1}) V^{0}(\tau_{i} \!+\!\! 1, \!\tau'_{j} \!+\!\! 1) \!\geq\! p_{j,1} V^{0}(\tau_{i} \!+\!\! 1,\! 1) \!+\! (p_{j,1}\!-\!p_{i,1}) V^{0}(\tau''_i \!+\!\! 1, \!\tau'_{j} \!+\!\! 1).\!\!
    \end{align}
    From $p_{j} < 1$ and~\eqref{eq:appb case b'3 sup},~\eqref{eq:appb case b'3 1} is smaller than the following equation
    \begin{align}\label{eq:appb case b'3 sub}
        \hspace{-1.3cm} \gamma \left[p_{j,1} - p_{j,1} p_{i,1}\right] \left[ V^{0}(\tau''_i + 1, 1) + V^{0}(\tau_i + 1, \tau_j + 1) - V^{0}(\tau_i + 1, 1) - V^{0}(\tau''_i + 1, \tau_j + 1) \right] \!\!,
    \end{align}
    which is smaller than 0 from~\eqref{eq:sub_0} and $p_{j,1} > p_{j,1} p_{i,1}$.
    Thus, based on~\eqref{eq:appb case b'3 1} and~\eqref{eq:appb case b'3 sub}, it can be derived that 
    \begin{align}
        p_{j,1}  V^{1}(\bm{\tau}) + (p_{j,1} - p_{i,1})  V^{1}(\bm{\tau}^{\circ}) - p_{j,1}  V^{1}(\bm{\tau} \wedge \bm{\tau}^{\circ}) - (p_{j,1} - p_{i,1})  V^{1}(\bm{\tau} \vee \bm{\tau}^{\circ}) \leq 0,
    \end{align}
    which is exactly~\eqref{eq:sup&sub_1,b}.
    \item[(b'.4)] If ${\pi^{1}} (\bm{\tau}) = j$ and ${\pi^{1}} (\bm{\tau}^{\circ}) = i$, then based on the case \textbf{(a')} and $p_{j,1} \geq p_{i,1}$, we have
    \begin{align}
        p_{i,1} V^{1}(\bm{\tau}) + (p_{i,1} - p_{j,1}) V^{1}(\bm{\tau}^{\circ}) \leq p_{i,1} V^{1}(\bm{\tau} \wedge \bm{\tau}^{\circ}) + (p_{i,1} - p_{j,1}) V^{1}(\bm{\tau} \vee \bm{\tau}^{\circ}).
    \end{align}
    In addition, as Theorem~\ref{theo:2-s-1-c} is derived by the probabilistic supermodularity, the optimal policy of $V^{1}(\mathbf{s})$ can also be proved to have the AoI-state threshold property in terms of $\tau_{j}$ in this case. Thus, we obtain that ${\pi^{1}} (\bm{\tau} \vee \bm{\tau}^{\circ}) = j$ from ${\pi^{1}} (\bm{\tau}) = j$, which implies that
    \vspace{-0.1cm}
    \begin{equation}
            W\left( \left( \tau_{i}, \tau'_{j} \right), j, V^{0} \right) \geq W\left( \left( \tau_{i}, \tau'_{j} \right), i, V^{0} \right), \label{eq:Wj > Wi}
            \vspace{-0.2cm}
    \end{equation}
    \vspace{-0.2cm}
    based on~\eqref{ineq:V&W}. Using~\eqref{eq:W-function},~\eqref{eq:Wj > Wi} is derived to
    \begin{align}
        p_{j,1} V^{0}(\tau_{i}\!+\!1, 1) \!+\! (1\!-\!p_{j,1}) V^{0}(\tau_{i}\!+\!1, \tau'_{j}\!+\!1) \geq p_{i,1} V^{0}(1, \tau'_{j}\!+\!1) \!+\! (1\!-\!p_{i,1}) V^{0}(\tau_{i}\!+\!1, \tau'_{j}\!+\!1),\!\!
    \end{align}
    which is exactly
    \vspace{-0.2cm}
    \begin{align}
        p_{j,1} V^{0}(\tau_{i}+1, 1) + (p_{i,1}- p_{j,1}) V^{0}(\tau_{i}+1, \tau'_{j}+1) \geq p_{i,1} V^{0}(1, \tau'_{j}+1). \label{eq: counter1}
    \end{align}
    Similarly, we derive the following inequality based on~\eqref{eq:W-function},~\eqref{ineq:V&W} and ${\pi^{1}} (\bm{\tau}^{\circ}) = i$,
    \begin{align}
        p_{i,1} V^{0}(1, \tau'_{j}+1) \geq p_{j,1} V^{0}(\tau''_{i}+1, 1) + (p_{i,1}-p_{j,1}) V^{0}(\tau''_{i}+1, \tau'_{j}+1).\label{eq: counter2}
    \end{align}
    Then, based on~\eqref{eq: counter1} and~\eqref{eq: counter2}, we have
     \begin{align}
         \hspace{-0.7cm} p_{j,1} V^{0}(\tau_{i}\!+\!\!1, \!1) \!+\! (p_{j,1}\!-\!p_{i,1}) V^{0}(\tau''_{i}\!+\!\!1, \tau'_{j}\!+\!\!1) \!\geq\! p_{j,1} V^{0}(\tau''_{i}\!+\!\!1, \!1) \!+\! (p_{j,1}\!-\!p_{i,1}) V^{0}(\tau_{i}\!+\!\!1, \tau'_{j}\!+\!\!1), \!\!
    \end{align}
    which is opposite to the~\eqref{eq:sup_0}. Therefore, this case cannot exist at the 1st iteration. As the value function of the last iteration is probabilistic supermodular, this case can also be proved to be non-existence by the same method for all other iterations in the value iteration.
\end{itemize}

Thus, the submodularity and the probabilistic supermodularity of the value function $V^{0}(\mathbf{s})$ propagates through the Bellman operator $\mathsf{B}[\cdot]$ to the optimal value function $V^{*}(\mathbf{s})$.

\section{Proof of Lemma~\ref{lemma:much larger}} \label{proof: much larger}
Similar to the proof of Lemma~\ref{lemma:monotone}, we assume that the initial value function $V^{0}(\mathbf{s})$ has the following property given states $\mathbf{s} = (\bm{\tau}, \mathbf{H})$ and $ \mathbf{s}' = (\bm{\tau}'_{(i)}, \mathbf{H})$ where $\tau'_{i} \gg \tau_{i}$,
\begin{equation}\label{eq: much larger 0}
    V^{0}\left(\mathbf{s}'\right) \ll V^{0} \left(\mathbf{s}\right).
\end{equation}
To prove Lemma~\ref{lemma:much larger} based on~\eqref{eq: much larger 0} and Lemma~\ref{lemma:converge of V*}, it is sufficient to show that the value function~$V^{1}(\mathbf{s})$ has the same property, i.e.
\begin{equation}\label{eq: much larger 1}
    V^{1}\left(\mathbf{s}'\right) \ll V^{1} \left(\mathbf{s}\right).
\end{equation}

Similar to the proof of Lemma~\ref{lemma:probabilistic sup}, we write the states as $\mathbf{s} = \bm{\tau}$ and $ \mathbf{s}' = \bm{\tau}'_{(i)}$, and the action as ${\pi^{\tilde{t}}}(\mathbf{s}) = i$ to represent $a^{\tilde{t}}_{i}=1$ for the state $\mathbf{s}$.
Moreover, we drop the constant AoI states in different optimal value functions according to~\eqref{eq: W i -i}. For example, given states $\bm{\tau}  = (\tau_1, \dots, \tau_{i}, \dots, \tau_N)$ and $\bm{\tau}' = (\tau_1, \dots, \tau'_{i}, \dots, \tau_N)$, we write them as $\bm{\tau}  = (\tau_{i})$ and $\bm{\tau}'  = (\tau'_{i})$.
In the following, we prove~\eqref{eq: much larger 1} by cases with different optimal actions of the states.
\begin{itemize}
    \item[(a)] If ${\pi^{1}} \left(\bm{\tau}'_{(i)}\right) = i$, then 
    \begin{align}
        & \hspace{-0.7cm} V^{1} \left(\bm{\tau}' \right) - V^{1} \left(\bm{\tau} \right) \\
        & \hspace{-0.7cm} \leq W \left(\bm{\tau}', i, V^{0} \right) - W \left(\bm{\tau}, i, V^{0} \right) \\
        & \hspace{-0.7cm} = p_{i,1} \left[r(\tau'_{i}, \tau_{j}) - r(\tau_{i}, \tau_{j})\right] \\
        & \hspace{-0.7cm} \quad \!+\! \gamma \left[p_{i,1} V^{0}(1, \tau_{j}\!+\!\!1) \!+\! (1\!-\!p_{i,1}) V^{0}(\tau'_{i}\!+\!\!1, \tau_{j}\!+\!\!1) \!-\! p_{i,1} V^{0}(1, \tau_{j}\!+\!\!1) \!-\! (1\!-\!p_{i,1}) V^{0}(\tau_{i}\!+\!\!1, \tau_{j}\!+\!\!1) \right] \!\! \\
        & \hspace{-0.7cm} \ll 0,
    \end{align}
    where the last inequality is base on $r(\tau'_{i}, \tau_{j}) \ll r(\tau_{i}, \tau_{j})$ and~\eqref{eq: much larger 0}.
    \item[(b)] If ${\pi^{1}} \left(\bm{\tau}'_{(i)}\right) = j$, then the proof is similar to the case (a) by showing $W \left(\bm{\tau}', j, V^{0} \right) - W \left(\bm{\tau}, j, V^{0} \right) \ll 0$.
\end{itemize}

\section{Proof of Lemma~\ref{lemma: probabilistic sup, N>2}} \label{proof: sup and sub, N>2}
Similar to the proof of Lemma~\ref{lemma:monotone}, we assume that the initial value function $V^{0}(\mathbf{s})$ is probabilistic supermodular given states $\mathbf{s} = (\bm{\tau}, \mathbf{H})$ and $\mathbf{s}^{\circ} = (\bm{\tau}^\circ, \mathbf{H})$, where $\bm{\tau}^\circ = (\tau''_i, \tau'_{j})$ with $ \tau''_i \ll \tau_{i}$ and $\tau'_{j} \geq \tau_{j}$,
\begin{align}
    p_{j,1} V^{0}\left(\mathbf{s} \wedge \mathbf{s}^{\circ}\right) + (p_{j,1} - p_{i,1}) V^{0}\left(\mathbf{s} \vee \mathbf{s}^{\circ}\right) \geq p_{j,1} V^{0}\left(\mathbf{s}\right) + (p_{j,1} - p_{i,1}) V^{0}\left(\mathbf{s}^{\circ}\right).\label{eq:sup_0 N>2} 
\end{align}
To prove Lemma~\ref{lemma: probabilistic sup, N>2} based on~\eqref{eq:sup_0 N>2} and Lemma~\ref{lemma:converge of V*}, it is sufficient to show that the value function~$V^{1}(\mathbf{s})$ is probabilistic supermodular, i.e.
\begin{equation}
    p_{j,1} V^{1}\left(\mathbf{s} \wedge \mathbf{s}^{\circ}\right) + (p_{j,1} - p_{i,1}) V^{1}\left(\mathbf{s} \vee \mathbf{s}^{\circ}\right) \geq p_{j,1} V^{1}\left(\mathbf{s}\right) + (p_{j,1} - p_{i,1}) V^{1}\left(\mathbf{s}^{\circ}\right) .\label{eq:sup_1 N>2}
\end{equation}
Similar to the proof of Lemma~\ref{lemma:much larger}, we write the states as $\mathbf{s} = \bm{\tau} = (\tau_{i}, \tau_{j})$ and $ \mathbf{s}^{\circ} = \bm{\tau}^{\circ} = (\tau''_{i}, \tau'_{j})$, and the action as ${\pi^{\tilde{t}}}(\mathbf{s}) = i$ to represent $a^{\tilde{t}}_{i}=1$, then $\bm{\tau} \vee \bm{\tau}^{\circ} = (\tau_{i}, \tau'_{j})$ and $\bm{\tau} \wedge \bm{\tau}^{\circ} = (\tau''_{i}, \tau_{j})$.
In the following, we prove~\eqref{eq:sup_1 N>2} by case \textbf{(a)} and \textbf{(b)} with different packet success rates.
\begin{itemize}
    \item[\textbf{(a)}] If $p_{j,1} \leq p_{i,1}$, then
    \begin{align}
        & p_{j,1} V^{1}(\bm{\tau}) + (p_{j,1} - p_{i,1}) V^{1}(\bm{\tau}^{\circ}) - p_{j,1} V^{1}(\bm{\tau} \wedge \bm{\tau}^{\circ}) - (p_{j,1} - p_{i,1}) V^{1}(\bm{\tau} \vee \bm{\tau}^{\circ}) \\
        & = p_{j,1} [V^{1}(\bm{\tau}) - V^{1}(\bm{\tau} \wedge \bm{\tau}^{\circ})] + (p_{i,1} - p_{j,1}) [V^{1}(\bm{\tau} \vee \bm{\tau}^{\circ}) - V^{1}(\bm{\tau}^{\circ})] \\
        & \leq 0,
    \end{align}
    where the last inequality is derived based on the proof of Lemma~\ref{lemma:monotone}.
    \item[\textbf{(b)}] If $p_{j,1} \geq p_{i,1}$, then there are ten cases with different optimal actions of the states. 
    
    (b.1) ${\pi^{1}}(\bm{\tau}) = {\pi^{1}}(\bm{\tau}^{\circ}) = i$, 
    (b.2) ${\pi^{1}}(\bm{\tau}) = {\pi^{1}}(\bm{\tau}^{\circ}) = j$, 
    (b.3) ${\pi^{1}}(\bm{\tau}) = i$ and ${\pi^{1}}(\bm{\tau}^{\circ}) = j$, 
    (b.4) ${\pi^{1}}(\bm{\tau}) = j$ and ${\pi^{1}}(\bm{\tau}^{\circ}) = i$, 
    (b.5) ${\pi^{1}}(\bm{\tau}) = {\pi^{1}}(\bm{\tau}^{\circ}) = k, (k \neq i,j)$ 
    (b.6) ${\pi^{1}}(\bm{\tau}) = i$ and ${\pi^{1}}(\bm{\tau}^{\circ}) = k$, 
    (b.7) ${\pi^{1}}(\bm{\tau}) = k$ and ${\pi^{1}}(\bm{\tau}^{\circ}) = i$, 
    (b.8) ${\pi^{1}}(\bm{\tau}) = j$ and ${\pi^{1}}(\bm{\tau}^{\circ}) = k$, 
    (b.9) ${\pi^{1}}(\bm{\tau}) = k$ and ${\pi^{1}}(\bm{\tau}^{\circ}) = j$, 
    (b.10) ${\pi^{1}}(\bm{\tau}) = k_1$ and ${\pi^{1}}(\bm{\tau}^{\circ}) = k_2$.
    
    For the cases (b.1), (b.2), (b.3), and (b.4), the proof is the same as the proof of Lemma~\ref{lemma:probabilistic sup}, because the states, except for sensor $i$ and sensor $j$, are constant in these cases. Therefore, in the sequel, we only prove the other cases.
    In addition, for the cases with ${\pi^{1}}(\bm{\tau}) = k$ or ${\pi^{1}}(\bm{\tau}^{\circ}) = k$, the states are written as $\bm{\tau} = (\tau_{i}, \tau_{j}, \tau_{k})$, as the AoI state $\tau_{k}$ is not constant.
    \item[(b.5)] If ${\pi^{1}}(\bm{\tau}) = {\pi^{1}}(\bm{\tau}^{\circ}) = k, (k \neq i,j)$, then 
    \begin{align}
        & \hspace{-0.5cm} p_{j,1} V^{1}(\bm{\tau}) + (p_{j,1}-p_{i,1}) V^{1}(\bm{\tau}^{\circ}) - p_{j,1} V^{1}(\bm{\tau} \wedge \bm{\tau}^{\circ}) - (p_{j,1}-p_{i,1}) V^{1}(\bm{\tau} \vee \bm{\tau}^{\circ}) \\
        & \hspace{-0.5cm} \leq \! p_{j,1} W(\bm{\tau}\!, k, \!V^{0}) \!+\! (p_{j,1}\!-\!p_{i,1}) W(\bm{\tau}^{\circ}\!\!, k, \!V^{0}) \!-\! p_{j,1} W(\bm{\tau} \!\!\wedge\! \bm{\tau}^{\circ}\!\!, k, \!V^{0}) \!-\! (p_{j,1}\!-\!p_{i,1}) W(\bm{\tau} \!\vee\! \bm{\tau}^{\circ}\!\!, k, \!V^{0}) \!\! \\
        & \hspace{-0.5cm} \leq \gamma p_{j,1} \!\left[ p_{k,1} V^{0}(\tau_{i}\!+\!1, \tau_{j}\!+\!1, 1) \!+\! (1\!-\!p_{k,1}) V^{0}(\tau_{i}\!+\!1, \tau_{j}\!+\!1, \tau_k\!+\!1) \right] \\
        & \hspace{-0.5cm} \quad + \gamma (p_{j,1} - p_{i,1}) [p_{k,1} V^{0}(\tau''_i+1, \tau'_{j}+1, 1) + (1-p_{k,1}) V^{0}(\tau''_i+1, \tau'_{j}+1, \tau_k+1)] \\
        & \hspace{-0.5cm} \quad - \gamma p_{j,1} [p_{k,1} V^{0}(\tau''_i+1, \tau_{j}+1, 1) + (1-p_{k,1}) V^{0}(\tau''_i+1, \tau_{j}+1, \tau_k+1)] \\
        \label{eq:appd case b.5 1}
        & \hspace{-0.5cm} \quad - \gamma (p_{j,1} - p_{i,1}) [p_{k,1} V^{0}(\tau_{i}+1, \tau'_{j}+1, 1) + (1-p_{k,1}) V^{0}(\tau_{i}+1, \tau'_{j}+1, \tau_k+1)],
    \end{align}
    where the second inequality is based on $r(\bm{\tau} \!\vee\! \bm{\tau}^{\circ}) < r(\bm{\tau})$. Using~\eqref{eq:sup_0 N>2}, it can be derived that\vspace{-0.5cm}
    \begin{align}
        & \gamma p_{k,1} [p_{j,1} V^{0}(\tau_{i}+1, \tau_{j}+1, 1) + (p_{j,1} - p_{i,1}) V^{0}(\tau''_i+1, \tau'_{j}+1, 1) \\
        \label{eq:appd case b.5 2}
        & \quad - p_{j,1} V^{0}(\tau''_i+1, \tau_{j}+1, 1) - (p_{j,1} - p_{i,1}) V^{0}(\tau_{i}+1, \tau'_{j}+1, 1)] \leq 0 
    \end{align}
    and \vspace{-0.4cm}
    \begin{align}
        & \gamma (1-p_{k,1}) [p_{j,1} V^{0}(\tau_{i}+1, \tau_{j}+1, \tau_k+1) + (p_{j,1} - p_{i,1}) V^{0}(\tau''_i+1, \tau'_{j}+1, \tau_k+1) \\
        \label{eq:appd case b.5 3}
        & \quad - p_{j,1} V^{0}(\tau''_i+1, \tau_{j}+1, \tau_k+1) - (p_{j,1} - p_{i,1}) V^{0}(\tau_{i}+1, \tau'_{j}+1, \tau_k+1)] \leq 0.\vspace{-0.4cm}
    \end{align}
    Based on~\eqref{eq:appd case b.5 1},~\eqref{eq:appd case b.5 2}, and~\eqref{eq:appd case b.5 3}, we have \vspace{-0.2cm}
    \begin{align}
        p_{j,1} V^{1}(\bm{\tau}) + (p_{j,1}-p_{i,1}) V^{1}(\bm{\tau}^{\circ}) - p_{j,1} V^{1}(\bm{\tau} \wedge \bm{\tau}^{\circ}) - (p_{j,1}-p_{i,1}) V^{1}(\bm{\tau} \vee \bm{\tau}^{\circ}) \leq 0,\vspace{-0.2cm}
    \end{align}
    \vspace{-0.3cm}which is exactly~\eqref{eq:sup_1 N>2}.
    \item[(b.6)] If ${\pi^{1}}(\bm{\tau}) = i$ and ${\pi^{1}}(\bm{\tau}^{\circ}) = k$, then
    \begin{align}
        & \hspace{-0.5cm} p_{j,1} V^{1}(\bm{\tau}) + (p_{j,1}-p_{i,1}) V^{1}(\bm{\tau}^{\circ}) - p_{j,1} V^{1}(\bm{\tau} \wedge \bm{\tau}^{\circ}) - (p_{j,1}-p_{i,1}) V^{1}(\bm{\tau} \vee \bm{\tau}^{\circ}) \\
        & \hspace{-0.5cm} \leq p_{j,1} W(\bm{\tau}, i, \! V^{0}) \!+\! (p_{j,1}\!-\!p_{i,1}) W(\bm{\tau}^{\circ}\!\!, k, \!V^{0}) \!-\! p_{j,1} W(\bm{\tau} \!\wedge\! \bm{\tau}^{\circ}\!\!, k, \!V^{0}) \!-\! (p_{j,1}\!-\!p_{i,1}) W(\bm{\tau} \!\vee\! \bm{\tau}^{\circ}\!\!, i, \!V^{0}) \!\! \\
        & \hspace{-0.5cm} \leq \gamma p_{j,1} \!\left[ p_{i,1} V^{0}(1, \tau_{j}\!+\!1, \tau_k\!+\!1) \!+\! (1\!-\!p_{i,1}) V^{0}(\tau_{i}\!+\!1, \tau_{j}\!+\!1, \tau_k\!+\!1) \right] \\
        & \hspace{-0.5cm} \quad + \gamma (p_{j,1} - p_{i,1}) [p_{k,1} V^{0}(\tau''_i+1, \tau'_{j}+1, 1) + (1-p_{k,1}) V^{0}(\tau''_i+1, \tau'_{j}+1, \tau_k+1)] \\
        & \hspace{-0.5cm} \quad - \gamma p_{j,1} [p_{k,1} V^{0}(\tau''_i+1, \tau_{j}+1, 1) + (1-p_{k,1}) V^{0}(\tau''_i+1, \tau_{j}+1, \tau_k+1)] \\
        \label{eq:appd case b.6 1}
        & \hspace{-0.5cm} \quad - \gamma (p_{j,1} - p_{i,1}) [p_{i,1} V^{0}(1, \tau'_{j}+1, \tau_k+1) + (1-p_{i,1}) V^{0}(\tau_{i}+1, \tau'_{j}+1, \tau_k+1)].
    \end{align}
    Using Lemma~\ref{lemma:much larger} and $\tau_{i} \gg \tau''_{i}$, we have
    \begin{align}\label{eq:appd case b.6 2}
        \hspace{-1cm} V^{0}(\tau_{i}\!+\!1,\! \tau_{j}\!+\!1,\! \tau_{k}\!+\!1) \!\ll\! V^{0}(1, \!\tau_{j}\!+\!1, \!\tau_k\!+\!1) \!-\! V^{0}(\tau''_{i}\!+\!1,\! \tau_{j}\!+\! 1, \!1) \!-\! V^{0}(\tau''_i\!+\!1, \!\tau_{j}\!+\!1, \!\tau_k\!+\!1)
    \end{align}
    and 
    \begin{align}\label{eq:appd case b.6 3}
        \hspace{-1cm} V^{0}(\tau_{i}\!+\!1,\! \tau'_{j}\!+\!1, \!\tau_k\!+\!1) \!\ll\! V^{0}(1, \!\tau'_{j}\!+\!1, \!\tau_k\!+\!1) \!-\! V^{0}(\tau''_i\!+\!1, \!\tau'_{j}\!+\!1, \!1) \!-\! V^{0}(\tau''_i\!+\!1, \tau'_{j}\!+\!1, \!\tau_k\!+\!1).
    \end{align}
    From~\eqref{eq:appd case b.6 2},~\eqref{eq:appd case b.6 3}, $V(\mathbf{s}) < 0$, and $0<p_{i,1},p_{j,1}<1$, we can derive that $V^{0}(\tau_{i}\!+\!1,\! \tau_{j}\!+\!1,\! \tau_{k}\!+\!1)$ and $V^{0}(\tau_{i}\!+\!1,\! \tau'_{j}\!+\!1, \!\tau_k\!+\!1)$ are dominate elements in~\eqref{eq:appd case b.6 1}. Therefore, to prove~\eqref{eq:sup_1 N>2} from~\eqref{eq:appd case b.6 1}, we only to prove the following inequality
    \begin{align}
        &(1\!-\!p_{i,1}) [p_{j,1} V^{0}(\tau_{i}\!+\!1, \tau_{j}\!+\!1, \tau_k\!+\!1) + (p_{j,1} - p_{i,1}) V^{0}(\tau''_i+1, \tau'_{j}+1, \tau_k+1) \\
        & \quad - p_{i,1} V^{0}(\tau''_i+1, \tau_{j}+1, \tau_k+1) + (p_{j,1} - p_{i,1}) V^{0}(\tau_{i}+1, \tau'_{j}+1, \tau_k+1)] \leq 0,
    \end{align}
    which can be derived from~\eqref{eq:sup_0 N>2} directly.
    
    For cases (b.7), (b.8), (b.9), and (b.10), the proof is similar to the case (b.6)
\end{itemize}
Thus, the probabilistic supermodularity of the value function $V^{0}(\mathbf{s})$ propagates through the Bellman operator $\mathsf{B}[\cdot]$ to the optimal value function $V^{*}(\mathbf{s})$.

\bibliographystyle{IEEEtran}

\end{document}